\documentclass[12pt,a4paper]{article}
\usepackage{authblk}

\usepackage[margin=1in]{geometry}
\usepackage[T1]{fontenc}
\usepackage[utf8]{inputenc}
\usepackage{lmodern}
\usepackage{microtype}
\usepackage{amsmath,amssymb,amsfonts,mathtools,bm}
\usepackage{amsthm}
\usepackage{mathrsfs}
\usepackage{bbm}
\usepackage{graphicx}
\usepackage{booktabs}
\usepackage{array}
\usepackage{enumitem}
\usepackage{xcolor}
\usepackage{hyperref}
\usepackage[nameinlink,capitalise]{cleveref}
\usepackage{placeins}
\usepackage{setspace}

\newif\ifprojectbib
\IfFileExists{ref.bib}{\projectbibtrue}{\projectbibfalse}
\ifprojectbib\else
  \renewcommand{\cite}[1]{\unskip}
\fi
\hypersetup{
  colorlinks=true,
  citecolor=blue,
  linkcolor=blue,
  urlcolor=blue,
  breaklinks=true
}
\allowdisplaybreaks
\newcommand{\cH}{\mathcal H}
\newcommand{\cN}{\mathcal N}
\newcommand{\cM}{\mathcal M}
\newcommand{\cS}{\mathcal S}
\newcommand{\cI}{\mathcal I}
\newcommand{\cD}{\mathcal D}
\newcommand{\cT}{\mathcal T}
\newcommand{\cE}{\mathcal E}
\newcommand{\cF}{\mathcal F}
\newcommand{\cC}{\mathcal C}
\newcommand{\cP}{\mathcal P}
\newcommand{\cQ}{\mathcal Q}

\newcommand{\id}{\mathrm{id}}
\newcommand{\misc}{\mathrm{MISC}}
\newcommand{\disc}{\mathrm{DISC}}
\newcommand{\para}{\mathrm{Par}}
\newcommand{\seq}{\mathrm{Seq}}
\newcommand{\gen}{\mathrm{Gen}}
\newcommand{\fixAB}{A\prec B}
\newcommand{\fixBA}{B\prec A}
\newcommand{\ox}{\otimes}
\newcommand{\tr}{\operatorname{Tr}}
\newcommand{\diag}{\operatorname{diag}}

\newcommand{\ket}[1]{\lvert #1\rangle}
\newcommand{\bra}[1]{\langle #1\rvert}

\newcommand{\ketbra}[2]{\lvert #1\rangle\!\langle #2\rvert}

\newcommand{\norm}[1]{\lVert #1\rVert}

\newcommand{\vect}[1]{\bm{#1}}

\newtheorem{theorem}{Theorem}[section]
\newtheorem{proposition}[theorem]{Proposition}
\newtheorem{lemma}[theorem]{Lemma}
\newtheorem{corollary}[theorem]{Corollary}
\newtheorem{definition}[theorem]{Definition}
\theoremstyle{remark}

\title{Causal-Class Hierarchies in Coherence-Constrained Channel Transformation}

\author[1,2]{Lin Zhu}
\author[3]{Benchi Zhao}
\author[3]{Xuanqiang Zhao}
\author[2]{Ranyiliu Chen\thanks{chenranyiliu@quantumsc.cn}}
\author[1]{Xin Wang\thanks{felixxinwang@hkust-gz.edu.cn}}
\author[2]{Shenggen Zheng\thanks{zhengshenggen@quantumsc.cn}}

\affil[1]{Thrust of Artificial Intelligence, Information Hub, The Hong Kong University of Science and Technology (Guangzhou), Guangdong 511453, China}
\affil[2]{Quantum Science Center of Guangdong--Hong Kong--Macao Greater Bay Area, Shenzhen 518045, China}
\affil[3]{QICI Quantum Information and Computation Initiative, School of Computing and Data Science, The University of Hong Kong, Pokfulam Road, Hong Kong}
\date{}

\begin{document}
\maketitle

\begin{abstract}
Higher-order quantum transformations allow multiple channel uses to be combined through different causal architectures, from parallel and fixed-order sequential networks to general higher-order processes. Whether this causal freedom improves channel transformation when the higher-order operation is also constrained by a resource theory remains largely unexplored. We study this question in the dynamical resource theory of coherence using a unified semidefinite-programming framework. For two qubit amplitude-damping channels and the identity target, we prove a strict causal hierarchy at every nontrivial damping strength under both maximally incoherent superchannels (MISC) and dephasing-covariant incoherent superchannels (DISC). In contrast, mixed-Pauli channels admit a common teleportation simulation that transfers the channel dependence to Bell-diagonal program states prepared in parallel. The remaining processing can then be absorbed into a single quantum operational, so parallel, fixed-order sequential, and general higher-order strategies achieve the same optimal error for any target. These results identify free program-state parallelisation as a structural obstruction to causal enhancement.
\end{abstract}

\newpage
\section{Introduction}
The causal arrangement of noisy quantum channels determines which uses can exchange information and exploit an intermediate memory. Parallel protocols place the supplied channels in one layer, fixed-order sequential networks allow information to pass from an earlier call to a later one, and more general deterministic higher-order processes obey weaker causal constraints~\cite{Chiribella2008CircuitArchitecture,Chiribella2008MemoryEffects,Chiribella2009Networks,Perinotti2017CausalStructures,Oreshkov2012NoCausalOrder,BisioPerinotti2019HigherOrder}. A basic operational question is whether this causal freedom improves channel transformation when the surrounding higher-order processing is subject to a fixed resource constraint. Given the same resource channels, can a sequential network or a more general deterministic higher-order process approximate a target channel more accurately than parallel processing?

A resource constraint is essential for the known-channel transformation task considered here. If arbitrary surrounding processing were allowed, the supplied interfaces could simply be ignored and the known target channel implemented directly. We therefore work in the dynamical resource theory of coherence, in which classical channels are free and the allowed higher-order maps obey a common coherence restriction~\cite{Theurer2019QuantifyingOperations,GourWinter2019DynamicalResource,LiuWinter2019ChannelResource,LiuYuan2020OperationalChannels,Saxena2020DynamicalCoherence,GourScandolo2021DynamicalResources,RegulaTakagi2021OneShot,Yuan2020OneShotDynamical}. We compare maximally incoherent superchannels (MISC) with the smaller class of dephasing-covariant incoherent superchannels (DISC), while varying only the allowed causal architecture. This separates the role of causal organization from the choice of free-operation constraint. Since $\disc\subseteq\misc$, persistence of a causal advantage under DISC also tests whether that advantage survives a genuinely stricter coherence constraint.

Previous work has established that causal architecture can affect communication, discrimination, metrology, and transformations of unknown operations~\cite{Chiribella2012PerfectDiscrimination,Bavaresco2021Hierarchy,Bavaresco2022UnitaryDiscrimination,Ebler2018Communication,Chiribella2021ZeroCapacity,zhao2020quantum,Liu2023Metrology,Quintino2019Reversing,QuintinoEbler2022UnitaryTransformations}. Resource-theoretic approaches have also quantified causal connection and the communication power permitted by different higher-order signalling structures~\cite{Milz2022CausalConnection,zhao2025communication}. In particular, Ref.~\cite{zhao2025communication} showed that causal architecture can affect channel-conversion performance in a resource theory of signalling, where replacement channels are free and the allowed higher-order operations are signalling-non-generating. For dynamical coherence, however, this does not determine either the achievable improvement or the channel structures that preclude it. An advantage for one free-operation class need not survive a stricter class, and an example of separation does not explain when all causal classes have the same transformation power. Existing channel resource theories typically optimize over a free-superchannel class without resolving parallel, fixed-order sequential, and more general causal subclasses~\cite{Gour2019Comparison,LiuWinter2019ChannelResource,LiuYuan2020OperationalChannels,GourWinter2019DynamicalResource,GourScandolo2021DynamicalResources}, while task-specific causal-order advantages do not by themselves answer this resource-constrained conversion question.

Here we establish both a strict causal hierarchy and a complementary sufficient condition for its collapse within the same coherence-constrained transformation problem. Two identical qubit amplitude-damping channels exhibit a strict parallel--sequential--general hierarchy for approximation of the identity at every nontrivial damping strength, and the optimal errors coincide between MISC and DISC for each causal class considered. By contrast, if the supplied resource channels admit free simulations from program states that can all be prepared in parallel, every general strategy can be matched on those resources by a parallel one. Qubit mixed-Pauli channels satisfy this condition through teleportation simulation, so the collapse holds for any finite collection of such channels and any target. The two regimes show that the usefulness of causal freedom depends on structural properties of the resource channels even when the resource-theoretic constraint is held fixed.

Our main results are organized around these two scientific findings. First, for two identical qubit amplitude-damping channels $\cN_{\rm AD}^{\varepsilon}$ and the identity target, we prove that for every $0<\varepsilon<1$ and every $\mathbf F\in\{\misc,\disc\}$,
\begin{equation}
D_{\mathbf F}^{\para}\left(\cN_{\rm AD}^{\varepsilon},\cN_{\rm AD}^{\varepsilon};\cI\right)
>
D_{\mathbf F}^{\seq}\left(\cN_{\rm AD}^{\varepsilon},\cN_{\rm AD}^{\varepsilon};\cI\right)
>
D_{\mathbf F}^{\gen}\left(\cN_{\rm AD}^{\varepsilon},\cN_{\rm AD}^{\varepsilon};\cI\right).
\label{eq:intro-hierarchy}
\end{equation}
For each of the three causal classes, the MISC and DISC optima are equal in this task. Exchange symmetry makes the two fixed orders equivalent for the identical resources. We determine the fixed-order optimum through a four-dimensional symmetric eigenvalue problem, while the general-process error is available in closed form,
\begin{equation}
D_{\mathbf F}^{\gen}\left(\cN_{\rm AD}^{\varepsilon},\cN_{\rm AD}^{\varepsilon};\cI\right)
=
\frac12\left[1-\sqrt{1-\varepsilon}\left(\varepsilon+\sqrt{\varepsilon^2+1-\varepsilon}\right)\right],
\qquad
\mathbf F\in\{\misc,\disc\}.
\label{eq:intro-ico-formula}
\end{equation}
This is a separation of the specified process cones; it does not identify the optimizing general process with a quantum switch or another particular quantum-controlled-order implementation.

Second, we prove a structural collapse theorem. If one use of each resource channel can prepare a channel-dependent program state, a resource-independent processor reconstructs the channel action from that state, and the induced simulators are free, then all program states can be prepared in one parallel layer and every subsequent use of the resource channels can be absorbed into a single residual quantum channel. Consequently, every general free strategy can be reproduced by a parallel free strategy on the given resource tuple. Qubit mixed-Pauli channels satisfy this condition through teleportation simulation. Hence, for any finite collection of such channels, any target channel, and either MISC or DISC, parallel processing, every fixed order, and the general-process class have the same optimal transformation error. This is a sufficient condition for collapse, not a necessary characterization of all channel families with no causal advantage.

The unified semidefinite-programming framework supports these two findings rather than constituting the primary conclusion. A common Choi-space formulation imposes the causal normalization and MISC or DISC conditions independently. For the identity target, an output twirl reduces the optimization to a single coherence-transfer coefficient $x_c$. Explicit feasible constructions, reduced primal--dual programs, and analytic upper-bound certificates then establish the amplitude-damping hierarchy and the program-state collapse on the same footing.

\paragraph{Structure of the paper.}
\Cref{sec:framework} fixes the channel, causal-class, and free-operation conventions and gives the unified SDP. \Cref{sec:coherence-reduction} derives the identity-target symmetry reduction. The amplitude-damping hierarchy and its analytic certificates are developed in \cref{sec:amplitude-damping}. \Cref{sec:collapse} proves the program-state parallelization theorem and applies it to mixed-Pauli channels, while \cref{sec:disc-robustness} establishes robustness under DISC. Numerical evidence beyond the analytic examples is reported in \cref{sec:numerics}, and \cref{sec:discussion} summarizes the scope and open questions. Detailed causal constraints, feasible constructions, and certificate calculations are collected in the appendices.

\section{Operational setting and incoherent higher-order transformations}
\label{sec:framework}

A fair causal comparison requires the resource-theoretic and causal restrictions to be specified independently. This section fixes the Choi conventions, defines the three causal classes used in the optimization, and then imposes MISC or DISC as a separate constraint on the same higher-order map.

\subsection{Channels, superchannels, and causal classes}
Let $\cH_X$ be the finite-dimensional Hilbert space associated with a quantum system $X$, let $d_X:=\dim\cH_X$, and denote by $\mathsf L(X)$ the space of linear operators on $\cH_X$. A quantum channel $\cN:\mathsf L(X_i)\to\mathsf L(X_o)$ is a completely positive trace-preserving linear map~\cite{nielsen2010quantum,watrous2018theory}. For brevity, we write $\cN:X_i\to X_o$ when the underlying operator spaces are clear from context. For a linear map $\cN:X_i\to X_o$, we use the unnormalised Choi operator~\cite{Jamiolkowski1972LinearTransformations,Choi1975CompletelyPositive}

\begin{equation}
J_{\cN}
:=
\sum_{i,j=0}^{d_{X_i}-1}
\ketbra{i}{j}_{X_i}
\ox
\cN(\ketbra{i}{j})_{X_o},
\qquad
\tr_{X_o}J_{\cN}=I_{X_i}.
\label{eq:choi-definition}
\end{equation}
A linear map $\cN:X_i\to X_o$ is completely positive if and only if $J_{\cN}\geq0$, and it is trace preserving if and only if $\tr_{X_o}J_{\cN}=I_{X_i}$. For two resource channels $\cN_A:A_i\to A_o$ and $\cN_B:B_i\to B_o$, a two-input superchannel $\cS:AB\to C$ maps them to an output channel $\cM_C=\cS(\cN_A,\cN_B)$~\cite{Chiribella2008Supermaps,BisioPerinotti2019HigherOrder}. Its Choi operator $\Theta^{\cS}$ acts through the link product~\cite{Chiribella2008CircuitArchitecture,Chiribella2009Networks},
\begin{equation}
J_{\cS(\cN_A,\cN_B)}
=
(J_{\cN_A}\ox J_{\cN_B})*\Theta^{\cS}.
\label{eq:superchannel-link}
\end{equation}
The link product with symbol $*$ is the Choi-space operation that connects compatible input and output systems of two quantum processes. The link-product convention and its explicit definition are given in \cref{app:notation}.

Throughout this work, a superchannel is a deterministic higher-order transformation: it maps every admissible tuple of input channels to an output channel and remains completely positive under arbitrary ancillary extensions. In finite dimensions this is encoded by a positive higher-order Choi operator, $\Theta^{\cS}\succeq0$, together with the appropriate normalization conditions~\cite{Chiribella2008Supermaps,Chiribella2009Networks,BisioPerinotti2019HigherOrder}. The causal classes differ only through those normalization conditions.

We write $A\prec B$ when channel $A$ is called before $B$, and $B\prec A$ for the reverse order. The comparison uses four classes: the parallel class $\para$, the fixed orders $\seq_{\fixAB}$ and $\seq_{\fixBA}$, and the general class $\gen$. A parallel network places both calls in a single layer; a fixed-order comb can pass a quantum memory from the first call to the second. The class $\gen$ is the full deterministic two-input higher-order feasible set defined by the normalization constraints used in this work. It contains the preceding classes and may contain causally nonseparable processes. We use the neutral symbol $\gen$ because membership in this cone, by itself, does not certify realization by a quantum switch or another quantum-controlled-order circuit. The classes are summarized in \cref{tab:causal-classes}, with inclusions
\begin{equation}
\para\subseteq\seq_{\fixAB}\subseteq\gen,
\qquad
\para\subseteq\seq_{\fixBA}\subseteq\gen.
\label{eq:causal-inclusions}
\end{equation}
When the two resource channels are identical, the fixed orders
$A\prec B$ and $B\prec A$ have the same optimal error by exchange
symmetry, and we therefore simply write $\seq$ when there are two identical resource channels. For nonidentical resources, the order is always displayed explicitly.

Operationally, the distinction is whether information may pass between the two channel calls. Parallel processing forbids such feed-forward, whereas a fixed-order comb carries an intermediate quantum memory. The general class relaxes the fixed-order comb conditions to the full two-input constraints used here. The three architectures are illustrated schematically in \cref{fig:causal-order}. In every case the resource channels and the coherence restriction are held fixed, so any difference in $D_{\mathbf F}^{\mathbf H}$ is attributable to the causal class itself~\cite{Chiribella2008MemoryEffects,Perinotti2017CausalStructures,OreshkovGiarmatzi2016CausalProcesses}.

\begin{table}[h!]
\centering
\caption{Causal classes compared in the two-use transformation problem. The last row denotes the general deterministic higher-order feasible set used in the SDP; no circuit realization is assumed for an arbitrary element of this class.}
\label{tab:causal-classes}
\begin{tabular}{p{0.16\linewidth}p{0.72\linewidth}}
\toprule
Class & Allowed organization of the two channel calls \\
\midrule
$\para$ & Both resource channels are used in one layer between a joint preparation and a joint post-processing stage. \\
$\seq_{A\prec B}$ & Channel $A$ is used before $B$, with an intermediate quantum memory. \\
$\seq_{B\prec A}$ & Channel $B$ is used before $A$, with an intermediate quantum memory. \\
$\gen$ & The full deterministic two-input higher-order class defined by the normalization constraints in \cref{app:general-sdp}; it contains the preceding classes and may include processes with no definite order. \\
\bottomrule
\end{tabular}
\end{table}

\begin{figure}[h!]
\centering
\IfFileExists{figures/Causal_order.png}{%
  \includegraphics[width=0.98\linewidth]{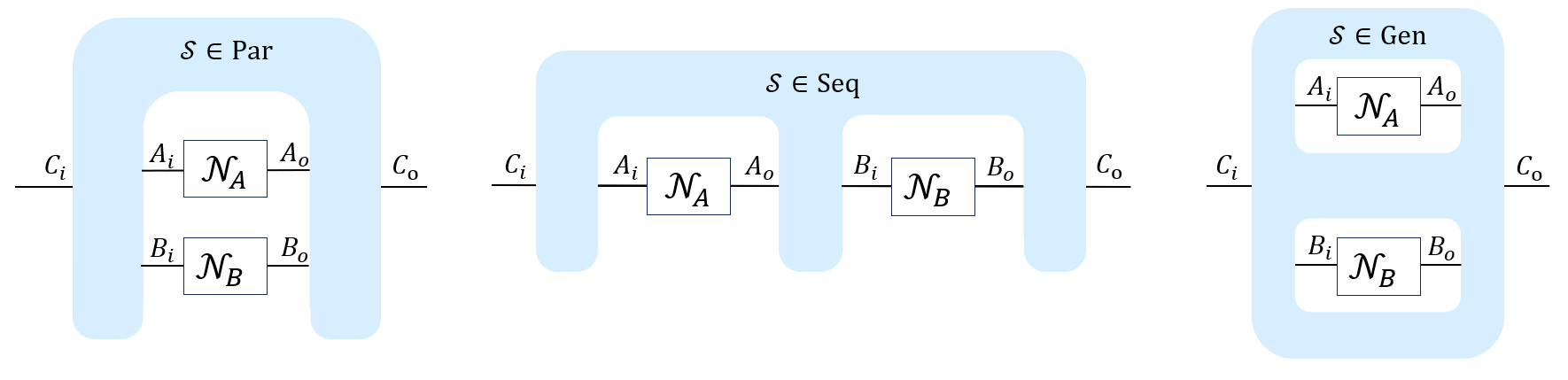}%
}{%
  \fbox{\parbox[c][5.0cm][c]{0.9\linewidth}{\centering
  Insert \texttt{figures/Causal\_order.png}: parallel, fixed-order sequential, and general higher-order architectures.}}%
}
\caption{Schematic causal architectures for two-use channel transformation. Parallel protocols use both resource channels in one layer, while fixed-order sequential protocols connect the two calls through an intermediate memory. The class $\gen$ denotes the general deterministic two-input higher-order feasible set used in this work; the schematic does not assert a particular circuit realization for an arbitrary element of $\gen$.}
\label{fig:causal-order}
\end{figure}

\subsection{Incoherent superchannels: MISC and DISC}
\label{subsec:misc-disc}

The free-operation constraint is defined relative to fixed incoherent bases on every input and output system. For a system $X$ with Hilbert space $\cH_X$, complete dephasing is the self-adjoint linear map
\begin{equation}
    \cD_X(M)
    :=
    \sum_{j=0}^{d_X-1}
    \ketbra{j}{j}
    M
    \ketbra{j}{j}.
    \label{eq:dephasing}
\end{equation}
For a channel from input system $X_i$ to output system $X_o$, channel dephasing simply removes the off-diagonal entries of its Choi operator in these fixed bases. We denote this linear supermap by $\Delta_X$:
\begin{equation}
    J_{\Delta_X(\cE)}
    :=
    \cD_X(J_{\cE}),
    \qquad
    \cD_X
    :=
    \cD_{X_i}\ox\cD_{X_o}.
    \label{eq:channel-dephasing-supermap}
\end{equation}
A channel $\cC:X_i\to X_o$ is classical when
\begin{equation}
    \Delta_X(\cC)=\cC,
    \label{eq:classical-channel-fixed-point}
\end{equation}
or equivalently when $J_{\cC}$ is diagonal in the fixed input--output basis. We denote the set of classical channels on $X$ by
$\mathscr C_X$.

For channel systems $X_1,\ldots,X_n$, write
\begin{equation}
    \boldsymbol{\Delta}_{\vect X}
    :=
    \left(
        \Delta_{X_1},\ldots,\Delta_{X_n}
    \right).
    \label{eq:tuple-channel-dephasing}
\end{equation} 
The one-input case reduces to the standard MISC and DISC definitions of dynamical coherence~\cite{Saxena2020DynamicalCoherence}. For several input channels, we impose the same relations simultaneously on all input dephasing maps.

\begin{definition}[Multi-input MISC and DISC]
\label{def:multi-input-misc-disc}
Let
\begin{equation}
    \cS:X_1\cdots X_n\to C
\end{equation}
be an $n$-input superchannel.

The map $\cS$ is a maximally incoherent superchannel, abbreviated MISC,
when
\begin{equation}
    \Delta_C
    \circ
    \cS
    \circ
    \boldsymbol{\Delta}_{\vect X}
    =
    \cS
    \circ
    \boldsymbol{\Delta}_{\vect X},
    \label{eq:misc-operational-definition}
\end{equation}
where the equality holds on the full linear space of input maps.

The map $\cS$ is a dephasing-covariant incoherent superchannel,
abbreviated DISC, when
\begin{equation}
    \Delta_C\circ\cS
    =
    \cS\circ\boldsymbol{\Delta}_{\vect X},
    \label{eq:disc-operational-definition}
\end{equation}
where the equality holds on the full linear space of input maps.
\end{definition}

These two conditions differ in where dephasing is required to act. Under MISC, dephasing every input channel is enough to make the output incoherent, so every tuple of classical channels is sent to a classical channel. DISC additionally requires output dephasing to commute with the higher-order transformation even for coherent inputs.

\begin{proposition}[Choi characterisation of multi-input MISC and DISC]
\label{prop:multi-input-choi-characterisation}
Let $\Theta^{\cS}$ be the Choi operator of an $n$-input higher-order
map
\begin{equation}
    \cS:X_1\cdots X_n\to C.
\end{equation}
Set
\begin{equation}
    \cD_{\vect X}
    :=
    \cD_{X_1}\ox\cdots\ox\cD_{X_n}.
    \label{eq:multi-input-dephasing}
\end{equation}
Then $\cS$ is MISC if and only if
\begin{equation}
    \left(
        \cD_{\vect X}\ox\cD_C
    \right)
    \left(
        \Theta^{\cS}
    \right)
    =
    \left(
        \cD_{\vect X}\ox\id_C
    \right)
    \left(
        \Theta^{\cS}
    \right).
    \label{eq:multi-input-misc-choi}
\end{equation}
It is DISC if and only if
\begin{equation}
    \left(
        \id_{\vect X}\ox\cD_C
    \right)
    \left(
        \Theta^{\cS}
    \right)
    =
    \left(
        \cD_{\vect X}\ox\id_C
    \right)
    \left(
        \Theta^{\cS}
    \right).
    \label{eq:multi-input-disc-choi}
\end{equation}
Consequently,
\begin{equation}
    \disc\subseteq\misc.
    \label{eq:disc-subset-misc}
\end{equation}
\end{proposition}

\begin{proof} 
Complete dephasing is self-adjoint with respect to the Hilbert--Schmidt inner product and commutes with transposition in the fixed incoherent basis. It can therefore be transferred between the two factors joined by a link product.

The Choi operators of the relevant composite higher-order maps are
\begin{equation}
    \Theta^{
        \cS\circ\boldsymbol{\Delta}_{\vect X}
    }
    =
    \left(
        \cD_{\vect X}\ox\id_C
    \right)
    \left(
        \Theta^{\cS}
    \right),
    \label{eq:choi-input-dephasing-composition}
\end{equation}
\begin{equation}
    \Theta^{
        \Delta_C\circ\cS
    }
    =
    \left(
        \id_{\vect X}\ox\cD_C
    \right)
    \left(
        \Theta^{\cS}
    \right),
    \label{eq:choi-output-dephasing-composition}
\end{equation}
and
\begin{equation}
    \Theta^{
        \Delta_C
        \circ
        \cS
        \circ
        \boldsymbol{\Delta}_{\vect X}
    }
    =
    \left(
        \cD_{\vect X}\ox\cD_C
    \right)
    \left(
        \Theta^{\cS}
    \right).
    \label{eq:choi-input-output-dephasing-composition}
\end{equation}
Equality of superchannels is equivalent to equality of their Choi operators. Hence Eq.~\eqref{eq:misc-operational-definition} is equivalent to Eq.~\eqref{eq:multi-input-misc-choi}, while Eq.~\eqref{eq:disc-operational-definition} is equivalent to Eq.~\eqref{eq:multi-input-disc-choi}.

Finally, if $\cS$ is DISC, then
\begin{equation}
\begin{aligned}
    \Delta_C
    \circ
    \cS
    \circ
    \boldsymbol{\Delta}_{\vect X}
    &=
    \cS
    \circ
    \boldsymbol{\Delta}_{\vect X}
    \circ
    \boldsymbol{\Delta}_{\vect X}\\
    &=
    \cS
    \circ
    \boldsymbol{\Delta}_{\vect X},
\end{aligned}
\end{equation}
where the second equality follows from idempotence of complete dephasing. Thus every DISC map is MISC.
\end{proof}

For the two-input setting used throughout this work, Eqs.~\eqref{eq:multi-input-misc-choi} and \eqref{eq:multi-input-disc-choi} reduce respectively to
\begin{equation}
    \cD_{ABC}
    \left(
        \Theta^{\cS}
    \right)
    =
    \left(
        \cD_{AB}\ox\id_C
    \right)
    \left(
        \Theta^{\cS}
    \right)
    \label{eq:misc-condition}
\end{equation}
and
\begin{equation}
    \left(
        \id_{AB}\ox\cD_C
    \right)
    \left(
        \Theta^{\cS}
    \right)
    =
    \left(
        \cD_{AB}\ox\id_C
    \right)
    \left(
        \Theta^{\cS}
    \right).
    \label{eq:disc-condition}
\end{equation}

\begin{lemma}[Classicality preservation]
\label{lem:misc-classicality-preservation}
Every multi-input MISC higher-order map sends each tuple of classical input channels to a classical output channel.
\end{lemma}

\begin{proof}
Let $\cC_k\in\mathscr C_{X_k}$ for every $k$. Then
\begin{equation}
    \Delta_{X_k}(\cC_k)=\cC_k.
    \label{eq:classical-input-fixed}
\end{equation}
Using the MISC condition,
\begin{equation}
\begin{aligned}
    \Delta_C
    \left[
        \cS(\cC_1,\ldots,\cC_n)
    \right]
    &=
    \left(
        \Delta_C
        \circ
        \cS
        \circ
        \boldsymbol{\Delta}_{\vect X}
    \right)
    (\cC_1,\ldots,\cC_n)\\
    &=
    \left(
        \cS
        \circ
        \boldsymbol{\Delta}_{\vect X}
    \right)
    (\cC_1,\ldots,\cC_n)\\
    &=
    \cS(\cC_1,\ldots,\cC_n).
\end{aligned}
\label{eq:misc-classicality-proof}
\end{equation}
Hence the output is fixed by $\Delta_C$ and is therefore classical.
\end{proof}

The converse of Lemma~\ref{lem:misc-classicality-preservation} is not used in this work. In particular, we do not infer the MISC Choi equation merely from preservation of normalised classical channels.

The program-state argument later composes a free higher-order map with free single-input simulators, so we record the corresponding closure property here.

\begin{lemma}[Closure under componentwise composition]
\label{lem:free-composition-closure}
Let
\begin{equation}
    \cS:A_1\cdots A_n\to C
\end{equation}
be an $n$-input higher-order map, and let
\begin{equation}
    \Phi_k:X_k\to A_k,
    \qquad
    k=1,\ldots,n,
\end{equation}
be single-input superchannels. Define their componentwise composition by
\begin{equation}
    \cP
    :=
    \cS
    \circ
    \left(
        \Phi_1,\ldots,\Phi_n
    \right).
    \label{eq:componentwise-free-composition}
\end{equation}
If $\cS$ and every $\Phi_k$ are MISC, then $\cP$ is MISC. If $\cS$ and every $\Phi_k$ are DISC, then $\cP$ is DISC.
\end{lemma}

\begin{proof}
Assume first that $\cS$ and every $\Phi_k$ are MISC. For every $k$, the one-input MISC condition gives
\begin{equation}
    \Phi_k\circ\Delta_{X_k}
    =
    \Delta_{A_k}
    \circ
    \Phi_k
    \circ
    \Delta_{X_k}.
    \label{eq:single-input-misc-identity}
\end{equation}
Therefore,
\begin{equation}
\begin{aligned}
    \cP
    \circ
    \boldsymbol{\Delta}_{\vect X}
    &=
    \cS
    \circ
    \left(
        \Phi_1\circ\Delta_{X_1},
        \ldots,
        \Phi_n\circ\Delta_{X_n}
    \right)\\
    &=
    \cS
    \circ
    \boldsymbol{\Delta}_{\vect A}
    \circ
    \left(
        \Phi_1\circ\Delta_{X_1},
        \ldots,
        \Phi_n\circ\Delta_{X_n}
    \right)\\
    &=
    \Delta_C
    \circ
    \cS
    \circ
    \boldsymbol{\Delta}_{\vect A}
    \circ
    \left(
        \Phi_1\circ\Delta_{X_1},
        \ldots,
        \Phi_n\circ\Delta_{X_n}
    \right)\\
    &=
    \Delta_C
    \circ
    \cP
    \circ
    \boldsymbol{\Delta}_{\vect X}.
\end{aligned}
\label{eq:misc-composition-closure-proof}
\end{equation}
Thus $\cP$ satisfies
Eq.~\eqref{eq:misc-operational-definition}.

Now assume that $\cS$ and every $\Phi_k$ are DISC. Then
\begin{equation}
    \Delta_{A_k}\circ\Phi_k
    =
    \Phi_k\circ\Delta_{X_k},
    \label{eq:single-input-disc-identity}
\end{equation}
and consequently
\begin{equation}
\begin{aligned}
    \Delta_C\circ\cP
    &=
    \Delta_C
    \circ
    \cS
    \circ
    \left(
        \Phi_1,\ldots,\Phi_n
    \right)\\
    &=
    \cS
    \circ
    \left(
        \Delta_{A_1}\circ\Phi_1,
        \ldots,
        \Delta_{A_n}\circ\Phi_n
    \right)\\
    &=
    \cS
    \circ
    \left(
        \Phi_1\circ\Delta_{X_1},
        \ldots,
        \Phi_n\circ\Delta_{X_n}
    \right)\\
    &=
    \cP
    \circ
    \boldsymbol{\Delta}_{\vect X}.
\end{aligned}
\label{eq:disc-composition-closure-proof}
\end{equation}
Hence $\cP$ is DISC.
\end{proof}

This lemma is the only composition rule needed for the program-state parallelization theorem below: each simulator must belong to the same free class as the higher-order map it replaces.

\subsection{Transformation error and unified SDP}
With the causal and free sets fixed, the remaining ingredient is the operational figure of merit. The channels entering the optimization are known and fixed, each resource interface is available once, and ancillas or intermediate memories may be used whenever the selected causal class allows them. MISC or DISC is imposed independently of that causal choice. Let $\cN_A$ and $\cN_B$ be two resource channels and let $\cM$ be the target channel. For a causal class $\mathbf H\in\{\para,\seq_{\fixAB},\seq_{\fixBA},\gen\}$ and a free class $\mathbf F\in\{\misc,\disc\}$, we define the optimal transformation error as 
\begin{equation} 
D_{\mathbf F}^{\mathbf H} \left( \cN_A,\cN_B;\cM \right) := \min_{\cS\in\mathbf H\cap\mathbf F} \frac12 \norm{ \cS(\cN_A,\cN_B)-\cM }_{\diamond}. 
\label{eq:transformation-error} 
\end{equation} 
The normalised diamond distance quantifies the worst-case distinguishability between the transformed channel and the target channel, including inputs entangled with an arbitrary reference system \cite{watrous2009semidefinite,watrous2018theory,regula2021operational}.

For the difference $\Phi:=\cE-\cF$ of two channels from $C_i$ to $C_o$,
we use the standard one-sided semidefinite representation of the
normalised diamond norm
\cite{watrous2009semidefinite,watrous2018theory},
\begin{equation}
    \frac12\norm{\Phi}_{\diamond}
    =
    \min_{Z,l}
    \left\{
        l:
        Z\succeq0,\quad
        J_{\Phi}\preceq Z,\quad
        \tr_{C_o}Z\preceq lI_{C_i}
    \right\}.
    \label{eq:one-sided-watrous}
\end{equation}

The optimization is an SDP because each ingredient is linear or semidefinite in Choi space: causal normalization, the MISC/DISC equations, and the one-sided Watrous representation of the diamond norm \cite{GutoskiWatrous2007,Chiribella2009Networks, watrous2009semidefinite,Bavaresco2021Hierarchy}. It can be written as
\begin{equation}
\begin{aligned}
    \textnormal{minimise}\quad
    &l\\
    \textnormal{subject to}\quad
    &\Theta^{\cS}\succeq0,\\
    &\Theta^{\cS}\in\mathrm{Comb}_{\mathbf H},\\
    &\Theta^{\cS}
    \text{ satisfies Eq.~\eqref{eq:misc-condition}
    or Eq.~\eqref{eq:disc-condition}},\\
    &Z\succeq0,\\
    &(J_{\cN_A}\ox J_{\cN_B})
    *
    \Theta^{\cS}
    -
    J_{\cM}
    \preceq Z,\\
    &\tr_{C_o}Z\preceq lI_{C_i}.
\end{aligned}
\label{eq:general-sdp}
\end{equation} 
For a fixed-order class, $\Theta^{\cS}$ obeys the comb equations of the chosen order. Thus the causal equations govern connectivity, while MISC or DISC governs coherence; neither constraint is used as a surrogate for the other. The explicit causal conditions and reduced primal--dual programs are collected in \cref{app:general-sdp,app:reduced-sdps}.

\section{Symmetry reduction and output coherence}
\label{sec:coherence-reduction}

For the analytic purification problem, both resource channels are qubits and the target is the identity. The output symmetries of this target collapse the relevant Choi space to three real directions. Define
\begin{equation}
\begin{aligned}
\Omega_0&:=\ketbra{00}{00}+\ketbra{11}{11},\\
\Omega_1&:=\ketbra{01}{01}+\ketbra{10}{10},\\
\Omega_c&:=\ketbra{00}{11}+\ketbra{11}{00},
\end{aligned}
\qquad
J_{\cI}=\Omega_0+\Omega_c.
\label{eq:omega-basis}
\end{equation}
\begin{lemma}[Identity-target output twirl]
\label{lem:identity-output-twirl}
Let
\begin{equation}
U_\phi:=\ketbra{0}{0}+e^{i\phi}\ketbra{1}{1},
\qquad
V_\phi:=\overline U_\phi{}_{C_i}\ox U_\phi{}_{C_o},
\qquad
V_X:=X_{C_i}\ox X_{C_o}.
\label{eq:output-twirl-generators}
\end{equation}
For the identity target, averaging a feasible pair $(\Theta,Z)$ over conjugation by $V_\phi$ for $\phi\in[0,2\pi)$ and by $V_X$ preserves feasibility and the objective for every causal and free class considered here. The averaged operators have the unique form
\begin{equation}
\Theta
=
T_0\ox\Omega_0+T_1\ox\Omega_1+T_c\ox\Omega_c,
\qquad
Z=z_0\Omega_0+z_1\Omega_1+z_c\Omega_c,
\label{eq:twirled-variables}
\end{equation}
with Hermitian $T_0,T_1,T_c$ and real $z_0,z_1,z_c$.
\end{lemma}

\begin{proof}
Conjugation by $V_\phi$ or $V_X$ is the Choi action induced by $\cN\mapsto\mathcal U\circ\cN\circ\mathcal U^\dagger$ on the output channel. Hence $J_{\cI}$ is fixed because
\begin{equation}
(\overline U\ox U)J_{\cI}(\overline U\ox U)^\dagger=J_{\cI}
\label{eq:identity-choi-invariance}
\end{equation}
for every unitary $U$. Unitary invariance of the diamond norm leaves the objective unchanged. The causal normalisation constraints are also preserved because they involve only partial traces and identity operators on $C_i$ and $C_o$. The generators are incoherent unitaries, so their conjugations commute with complete dephasing and preserve both free Choi equations. Positivity is preserved by unitary conjugation, and convexity permits averaging~\cite{VollbrechtWerner2001Symmetry,GatermannParrilo2004Symmetry,Vallentin2009SymmetrySDP}.

It remains to compute the fixed-point space. In the ordered basis $\ket{00},\ket{01},\ket{10},\ket{11}$ of $C_iC_o$, the phase action has charges $0,1,-1,0$. Averaging over $\phi$ removes all off-diagonal matrix elements except those between $\ket{00}$ and $\ket{11}$. Conjugation by $V_X$ exchanges $\ket{00}\leftrightarrow\ket{11}$ and $\ket{01}\leftrightarrow\ket{10}$. It therefore identifies the corresponding diagonal coefficients and removes the antisymmetric imaginary combination $i(\ketbra{00}{11}-\ketbra{11}{00})$. The Hermitian fixed-point space is exactly
\begin{equation}
\operatorname{span}_{\mathbb R}\{\Omega_0,\Omega_1,\Omega_c\}.
\label{eq:output-fixed-space}
\end{equation}
Applying this computation to the $C_iC_o$ legs of $\Theta$ and to $Z$ proves \eqref{eq:twirled-variables}.
\end{proof} Set
\begin{equation}
x_\alpha
:=
(J_{\cN_A}\ox J_{\cN_B})*T_\alpha,
\qquad
\alpha\in\{0,1,c\}.
\label{eq:x-alpha}
\end{equation}
Positivity gives $T_1\succeq0$ and $T_0\pm T_c\succeq0$. The reduced incoherence constraints on the blocks $T_0$, $T_1$, and $T_c$ are summarised in Table~\ref{tab:misc-disc-constraints}.
\begin{table}[h!]
\centering
\caption{Reduced incoherence constraints on the blocks $T_0$, $T_1$, and $T_c$ under MISC and DISC.}
\label{tab:misc-disc-constraints}
\begin{tabular}{c@{\qquad}ccc}
\toprule
& $T_0$ & $T_1$ & $T_c$ \\
\midrule
$\misc$
& unrestricted
& unrestricted
& $\cD_{AB}(T_c)=0$ \\[1mm]
$\disc$
& $T_0=\cD_{AB}(T_0)$
& $T_1=\cD_{AB}(T_1)$
& $\cD_{AB}(T_c)=0$ \\
\bottomrule
\end{tabular}
\end{table}

Thus DISC strengthens MISC only by dephasing the population blocks $T_0$ and $T_1$.

The three blocks play different roles. $T_0$ and $T_1$ govern the two population sectors, whereas $T_c$ multiplies $\Omega_c$ and is the only block that carries the phase relation between $\ket{0}$ and $\ket{1}$. After contraction with the resource Choi operators, its coefficient $x_c$ satisfies $\cS(\cN_A,\cN_B)(\ketbra{0}{1})=x_c\ketbra{0}{1}$. The diamond-distance problem can therefore be recast as a coherence-transfer problem, as the following proposition makes precise.

\begin{proposition}[Output-coherence reduction]
\label{prop:output-coherence-main}
Let $\cN_A$ and $\cN_B$ be qubit channels, let $\mathbf H\in\{\para,\seq_{\fixAB},\seq_{\fixBA},\gen\}$, and let $\mathbf F\in\{\misc,\disc\}$. For the identity target, an optimal reduced solution may be chosen with $T_1=0$. Its output Choi operator is
\begin{equation}
J_{\cS(\cN_A,\cN_B)}
=
\Omega_0+x_c\Omega_c,
\label{eq:output-dephasing-channel}
\end{equation}
and the channel acts on off-diagonal matrix units as $\cS(\cN_A,\cN_B)(\ketbra{0}{1})=x_c\ketbra{0}{1}$. Defining the auxiliary optimal coherence
\begin{equation}
\mu_{\mathbf F,\mathbf H}(\cN_A,\cN_B)
:=
\max_{\cS\in\mathbf H\cap\mathbf F}x_c,
\label{eq:mu-definition}
\end{equation}
one has
\begin{equation}
D_{\mathbf F}^{\mathbf H}(\cN_A,\cN_B;\cI)
=
\frac{1-\mu_{\mathbf F,\mathbf H}(\cN_A,\cN_B)}{2}.
\label{eq:error-coherence-relation}
\end{equation}
\end{proposition}

\begin{proof}
The twirled output is $J_{\mathrm{out}}=x_0\Omega_0+x_1\Omega_1+x_c\Omega_c$. Trace preservation gives $x_0+x_1=1$, while complete positivity gives $x_1\ge0$ and $x_0\pm x_c\ge0$. Hence the output is the Pauli channel with probabilities
\begin{equation}
\left(
\frac{x_0+x_c}{2},
\frac{x_1}{2},
\frac{x_1}{2},
\frac{x_0-x_c}{2}
\right)
\label{eq:pauli-probabilities}
\end{equation}
for $I,X,Y,Z$, respectively. Pauli channels are jointly teleportation covariant, and their diamond distance equals the $\ell_1$ distance of their Pauli probability vectors~\cite{sacchi2005optimal,watrous2018theory}. Therefore
\begin{equation}
\frac12\norm{\cS(\cN_A,\cN_B)-\cI}_\diamond
=
1-\frac{x_0+x_c}{2}
=
\frac{1-(x_c-x_1)}{2}.
\label{eq:distance-before-t1}
\end{equation}

For any feasible $(T_0,T_1,T_c)$, replace it by
\begin{equation}
(\widetilde T_0,\widetilde T_1,\widetilde T_c)
:=(T_0+T_1,0,T_c).
\label{eq:t1-removal}
\end{equation}
Because $T_1\succeq0$ and $T_0\pm T_c\succeq0$, one has
$\widetilde T_0\pm\widetilde T_c\succeq0$. The causal normalisation constraints depend only on $S:=T_0+T_1$ and therefore remain unchanged. The MISC condition is unchanged, and under DISC the sum $T_0+T_1$ remains diagonal. The replacement preserves $x_c$, sets $\widetilde x_1=0$, and gives $\widetilde x_0=1$, so it cannot increase \cref{eq:distance-before-t1}. The optimum is therefore $(1-x_c)/2$, proving \cref{eq:error-coherence-relation}.
\end{proof}

For an input qubit
\begin{equation}
\rho=
\begin{pmatrix}
p&c\\
c^*&1-p
\end{pmatrix},
\end{equation}
the optimal twirled output in \cref{eq:output-dephasing-channel} acts as
\begin{equation}
\rho\longmapsto
\begin{pmatrix}
p&x_c c\\
x_c c^*&1-p
\end{pmatrix}.
\label{eq:xc-bloch-action}
\end{equation}
The coefficient $x_c$ therefore has an immediate operational meaning: it rescales the off-diagonal matrix elements, or equivalently the transverse Bloch-vector components. The endpoints are $x_c=1$ for the identity and $x_c=0$ for complete dephasing. By \cref{eq:error-coherence-relation}, comparing causal classes is exactly the same as comparing how much of this coherence each admissible higher-order map can preserve.

For identical resource channels, exchanging $A\leftrightarrow B$ maps $\seq_{\fixAB}\cap\mathbf F$ bijectively onto $\seq_{\fixBA}\cap\mathbf F$ and preserves the target and objective. Hence the two fixed-order optima coincide. Throughout the amplitude-damping analysis, $D_{\mathbf F}^{\seq}$ and $\mu_{\mathbf F,\seq}$ denote this common fixed-order value.

\section{Strict hierarchy for amplitude-damping purification}
\label{sec:amplitude-damping}

Amplitude damping provides a simple nonunital test in which population loss and phase survival are tied to different Kraus branches. For $0\le\varepsilon\le1$, the qubit channel is described by~\cite{nielsen2010quantum}
\begin{equation}
K_0=\ketbra{0}{0}+\sqrt{1-\varepsilon}\ketbra{1}{1},
\qquad
K_1=\sqrt{\varepsilon}\ketbra{0}{1}.
\label{eq:ad-kraus}
\end{equation}
Here $K_0$ is the no-jump branch, which preserves the logical basis while attenuating coherence by $s=\sqrt{1-\varepsilon}$; $K_1$ transfers excited-state population to the ground state and carries no phase on its own. Two channel uses therefore provide distinct population and coherence pathways that a higher-order process can correlate. This makes amplitude damping a useful complement to existing channel-purification settings based on symmetry filters, virtual multi-copy mitigation, and higher-order processing~\cite{Leung1997ApproximateQEC,Fletcher2008StructuredQEC,Das2024QuantumFilters,Liu2025VirtualChannelPurification,Tsubouchi2025SymmetricChannelVerification,Niwa2026ScalingOptimalPurification,zhao2025communication}. We consider the two-use task $(\cN_{\rm AD}^{\varepsilon})^{\ox2}\to\cI$ and write
\begin{equation}
s:=\sqrt{1-\varepsilon},
\qquad
R_\varepsilon:=\sqrt{\varepsilon^2+1-\varepsilon}.
\label{eq:s-R-definitions}
\end{equation}
The analytic work is carried out under MISC; \cref{sec:disc-robustness} then shows that every optimum constructed here is also attainable under DISC.

\subsection{Exact fixed-order sequential error}

After the phase twirl, most fixed-order comb variables are irrelevant to the objective. The diagonal block $S$ distributes weight among three coherence-carrying pathways, while $T_c$ couples histories whose relative phase survives to the output. Positivity limits each coupling by the geometric mean of the corresponding populations. At saturation, the fixed-order SDP reduces to three nonnegative amplitudes $x,y,z$ on the unit sphere.

\begin{proposition}[Exact fixed-order sequential error]
\label{prop:sequential-main}
For every $0<\varepsilon<1$,
\begin{equation}
D_{\misc}^{\seq}
\left(
\cN_{\rm AD}^{\varepsilon},
\cN_{\rm AD}^{\varepsilon};
\cI
\right)
=
\frac12
\left[
1-
\max_{\substack{x,y,z\ge0\\x^2+y^2+z^2=1}}
F_\varepsilon(x,y,z)
\right],
\label{eq:sequential-variational}
\end{equation}
where
\begin{equation}
F_\varepsilon(x,y,z)
:=
2sx\sqrt{y^2+s^2z^2}
+2\varepsilon s yz
+\varepsilon s z^2.
\label{eq:F-epsilon}
\end{equation}
The same optimal error is obtained for the two fixed orders
$A\prec B$ and $B\prec A$.
\end{proposition}

\begin{proof}
By \cref{prop:output-coherence-main}, it is enough to determine the
maximum output-coherence coefficient
\begin{equation}
\mu_{\misc,\seq}
\left(
\cN_{\rm AD}^{\varepsilon},
\cN_{\rm AD}^{\varepsilon}
\right),
\label{eq:seq-mu-proof-definition}
\end{equation}
since
\begin{equation}
D_{\misc}^{\seq}
\left(
\cN_{\rm AD}^{\varepsilon},
\cN_{\rm AD}^{\varepsilon};
\cI
\right)
=
\frac{
1-
\mu_{\misc,\seq}
\left(
\cN_{\rm AD}^{\varepsilon},
\cN_{\rm AD}^{\varepsilon}
\right)
}{2}.
\label{eq:seq-error-mu-proof}
\end{equation}

For achievability, \cref{app:explicit-achievability} gives an explicit diagonal fixed-order comb and coherence block for every $x,y,z\ge0$ with $x^2+y^2+z^2=1$. The construction is MISC feasible, satisfies the $A\prec B$ comb constraints, and produces
\begin{equation}
    x_c=F_\varepsilon(x,y,z).
    \label{eq:seq-achievable-xc}
\end{equation}
Hence
\begin{equation}
\mu_{\misc,\seq}
\left(
\cN_{\rm AD}^{\varepsilon},
\cN_{\rm AD}^{\varepsilon}
\right)
\ge
\max_{\substack{x,y,z\ge0\\x^2+y^2+z^2=1}}
F_\varepsilon(x,y,z).
\label{eq:seq-achievability-bound}
\end{equation}

For the converse, phase twirling and complex conjugation reduce an
arbitrary fixed-order feasible point to the three real charge blocks
described in \cref{app:sequential-certificate}. Let $(x,y,z)$ maximise
$F_\varepsilon$ over $x,y,z\ge0$ with $x^2+y^2+z^2=1$, and define
\begin{equation}
\mu:=F_\varepsilon(x,y,z).
\label{eq:seq-proof-mu}
\end{equation}
The stationarity equations of this constrained maximisation yield the
positive certificate operators $P_\pm,Q_\pm,R_\pm$ constructed in
\cref{eq:seq-certificate-operators}. For every reduced fixed-order
feasible point, they satisfy
\begin{equation}
\begin{aligned}
\mu-x_c={}&
\langle P_+,S_{00}+C_{00}\rangle
+
\langle P_-,S_{00}-C_{00}\rangle\\
&+
\langle Q_+,S_{0,-1}+K(r)\rangle
+
\langle Q_-,S_{0,-1}-K(r)\rangle\\
&+
\langle R_+,S_{-1,0}+K(t)\rangle
+
\langle R_-,S_{-1,0}-K(t)\rangle.
\end{aligned}
\label{eq:seq-certificate-main}
\end{equation}
Every matrix appearing as the second argument of an inner product on the
right-hand side is positive semidefinite by primal feasibility, while
\cref{app:sequential-certificate} proves
\begin{equation}
P_\pm\succeq0,
\qquad
Q_\pm\succeq0,
\qquad
R_\pm\succeq0
\label{eq:seq-certificate-positive-main}
\end{equation}
for $0<\varepsilon<1$. Therefore
\begin{equation}
x_c\le\mu
=
\max_{\substack{x,y,z\ge0\\x^2+y^2+z^2=1}}
F_\varepsilon(x,y,z).
\label{eq:seq-converse-bound}
\end{equation}
Combining the achievable and converse bounds gives
\begin{equation}
\mu_{\misc,\seq}
\left(
\cN_{\rm AD}^{\varepsilon},
\cN_{\rm AD}^{\varepsilon}
\right)
=
\max_{\substack{x,y,z\ge0\\x^2+y^2+z^2=1}}
F_\varepsilon(x,y,z).
\label{eq:seq-mu-variational-proof}
\end{equation}
Substituting this identity into
\cref{eq:seq-error-mu-proof} proves
\cref{eq:sequential-variational}.

Finally, exchanging $A$ and $B$ maps the feasible set for the order
$A\prec B$ bijectively onto that for $B\prec A$ and leaves the two
identical resource channels and the identity target invariant.
Therefore the two fixed orders have the same optimal transformation
error.
\end{proof}

\begin{corollary}[Spectral expression for the fixed-order optimum]
\label{cor:sequential-spectral}
For $0<\varepsilon<1$, let $s=\sqrt{1-\varepsilon}$ and define the real symmetric matrix
\begin{equation}
H_\varepsilon
:=
\begin{pmatrix}
0 & 0 & s & 0\\
0 & 0 & 0 & s^2\\
s & 0 & 0 & \varepsilon s\\
0 & s^2 & \varepsilon s & \varepsilon s
\end{pmatrix}.
\label{eq:sequential-H-epsilon}
\end{equation}
Then
\begin{equation}
\mu_{\misc,\seq}
\left(
\cN_{\rm AD}^{\varepsilon},
\cN_{\rm AD}^{\varepsilon}
\right)
=
\lambda_{\max}(H_\varepsilon),
\label{eq:sequential-mu-eigenvalue}
\end{equation}
and hence
\begin{equation}
D_{\misc}^{\seq}
\left(
\cN_{\rm AD}^{\varepsilon},
\cN_{\rm AD}^{\varepsilon};
\cI
\right)
=
\frac{1-\lambda_{\max}(H_\varepsilon)}{2}.
\label{eq:sequential-error-eigenvalue}
\end{equation}
Equivalently, $\mu_{\misc,\seq}$ is the largest real root of
\begin{equation}
p_\varepsilon(\mu)
:=
(\mu^2-s^2)
(\mu^2-\varepsilon s\mu-s^4)
-
\varepsilon^2s^2\mu^2
=
0.
\label{eq:sequential-quartic}
\end{equation}
By \cref{prop:misc-disc-equality-main}, the same expressions also give the DISC fixed-order optimum.
\end{corollary}

\begin{proof}
For any $x,y,z\ge0$, Cauchy--Schwarz gives
\begin{equation}
x\sqrt{y^2+s^2z^2}
=
\max_{\substack{u,v\ge0\\u^2+v^2=x^2}}
\left(uy+svz\right).
\label{eq:sequential-cauchy-lift}
\end{equation}
If $w:=\sqrt{y^2+s^2z^2}>0$, equality is attained at
\begin{equation}
u=\frac{xy}{w},
\qquad
v=\frac{xsz}{w},
\label{eq:sequential-cauchy-equality}
\end{equation}
while for $w=0$ both sides of \cref{eq:sequential-cauchy-lift} vanish. Therefore the variational objective in \cref{eq:sequential-variational} satisfies
\begin{equation}
\begin{aligned}
&\max_{\substack{x,y,z\ge0\\x^2+y^2+z^2=1}}
\left[
2sx\sqrt{y^2+s^2z^2}
+2\varepsilon syz
+\varepsilon sz^2
\right]
\\
&\qquad=
\max_{\substack{u,v,y,z\ge0\\u^2+v^2+y^2+z^2=1}}
\left[
2suy
+2s^2vz
+2\varepsilon syz
+\varepsilon sz^2
\right].
\end{aligned}
\label{eq:sequential-lifted-rayleigh}
\end{equation}
Writing
\begin{equation}
q:=(u,v,y,z)^{\mathsf T},
\end{equation}
the last objective is $q^{\mathsf T}H_\varepsilon q$. Since $H_\varepsilon$ is real symmetric and entrywise nonnegative, replacing any real vector $q$ by its entrywise absolute value preserves its Euclidean norm and cannot decrease the quadratic form. Hence the nonnegativity constraint does not change the largest Rayleigh quotient, and
\begin{equation}
\max_{\substack{q\ge0\\q^{\mathsf T}q=1}}
q^{\mathsf T}H_\varepsilon q
=
\lambda_{\max}(H_\varepsilon).
\label{eq:sequential-rayleigh-max}
\end{equation}
Combining this equality with \cref{prop:sequential-main,eq:error-coherence-relation} proves \cref{eq:sequential-mu-eigenvalue,eq:sequential-error-eigenvalue}.

To derive the polynomial, write
\begin{equation}
H_\varepsilon
=
\begin{pmatrix}
0 & D\\
D & B
\end{pmatrix},
\qquad
D:=\diag(s,s^2),
\qquad
B:=
\begin{pmatrix}
0 & \varepsilon s\\
\varepsilon s & \varepsilon s
\end{pmatrix}.
\label{eq:sequential-H-blocks}
\end{equation}
For $\mu\neq0$, the Schur complement gives
\begin{equation}
\begin{aligned}
\det(\mu I_4-H_\varepsilon)
&=
\det\!\left(\mu^2I_2-\mu B-D^2\right)
\\
&=
(\mu^2-s^2)
(\mu^2-\varepsilon s\mu-s^4)
-\varepsilon^2s^2\mu^2.
\end{aligned}
\label{eq:sequential-characteristic-polynomial}
\end{equation}
Both sides are polynomials in $\mu$, so the identity extends to $\mu=0$. Because $H_\varepsilon$ is real symmetric, all of its eigenvalues are real; therefore $\lambda_{\max}(H_\varepsilon)$ is the largest real root of \cref{eq:sequential-quartic}.
\end{proof}

The variational formula also shows where the fixed-order gain comes from. The weights $x^2,y^2,z^2$ label three relevant pathways, and the square roots are the geometric-mean bounds imposed by positivity. The factor $s$ records the no-jump coherence that survives a channel use. The first term in $F_\varepsilon$ is memory-assisted coherence transfer; the terms proportional to $\varepsilon s$ use relaxation information available before the second call. The intermediate memory therefore enlarges the set of coherence pairings rather than merely changing the preparation or decoding.

\subsection{Parallel strategies are strictly weaker}

\begin{lemma}[Strict parallel-to-sequential separation]
\label{lem:no-parallel-saturation-main}
For every $0<\varepsilon<1$,
\begin{equation}
D_{\misc}^{\para}
\left(
\cN_{\rm AD}^{\varepsilon},
\cN_{\rm AD}^{\varepsilon};
\cI
\right)
>
D_{\misc}^{\seq}
\left(
\cN_{\rm AD}^{\varepsilon},
\cN_{\rm AD}^{\varepsilon};
\cI
\right).
\label{eq:strict-par-seq-error}
\end{equation}
\end{lemma}

\begin{proof}
By \cref{prop:output-coherence-main}, the claimed inequality is equivalent to
\begin{equation}
\mu_{\misc,\para}
\left(
\cN_{\rm AD}^{\varepsilon},
\cN_{\rm AD}^{\varepsilon}
\right)
<
\mu_{\misc,\seq}
\left(
\cN_{\rm AD}^{\varepsilon},
\cN_{\rm AD}^{\varepsilon}
\right).
\label{eq:strict-par-seq-mu}
\end{equation}
Suppose that a parallel feasible point attains the sequential optimum $\mu$. Since every parallel point is sequentially feasible, equality must hold in each nonnegative term of \cref{eq:seq-certificate-main}. Under parallel normalisation the relevant charge blocks reduce to
\begin{equation}
S_{00}=\diag(a,b,g,h),
\quad
S_{0,-1}=\diag(b,h),
\quad
S_{-1,0}=\diag(g,h),
\quad
a+b+g+h=1.
\label{eq:parallel-charge-blocks}
\end{equation}
The certificate has $P_-\succeq\Delta\ketbra{e_3}{e_3}$ with $\Delta>0$, so equality forces $g=0$. Positivity of $S_{-1,0}\pm K(t)$ then gives $t=0$, and the vanishing $R_+$ term gives $h=0$. Positivity of $S_{0,-1}\pm K(r)$ gives $r=0$, while the vanishing $Q_-$ term and its strictly positive coefficient force $b=0$. Hence $a=1$. Finally, $S_{00}\pm C_{00}\succeq0$ together with $\diag(C_{00})=0$ forces $C_{00}=0$, and therefore $x_c=0$.

This is impossible because the feasible point $x=y=1/\sqrt2$, $z=0$ in \cref{eq:sequential-variational} gives
\begin{equation}
\mu_{\misc,\seq}(\cN_{\rm AD}^{\varepsilon},\cN_{\rm AD}^{\varepsilon})
\ge s>0.
\label{eq:seq-positive-lower-bound}
\end{equation}
Thus no parallel point can attain the sequential optimum. The parallel feasible set is closed and trace normalised in a finite-dimensional space, so it is compact and the inequality is strict.
\end{proof}

The separation from parallel processing is encoded directly in the equality conditions of the dual certificate. Matching the sequential optimum would force every positive overlap in \cref{eq:seq-certificate-main} to vanish. Parallel normalization then propagates zero weights through the coherence-supporting blocks until $x_c=0$, contradicting the positive sequential value. In this precise sense, the advantage requires information carried through the intermediate memory from the first call to the second.

\subsection{Exact general-process transformation error}

\begin{proposition}[Exact general-process transformation error]
\label{prop:ico-main}
For every $0<\varepsilon<1$ with $R_\varepsilon:=\sqrt{\varepsilon^2+1-\varepsilon}$, 
\begin{equation}
D_{\misc}^{\gen}
\left(
\cN_{\rm AD}^{\varepsilon},
\cN_{\rm AD}^{\varepsilon};
\cI
\right)
=
\frac12
\left[
1-
s\left(
\varepsilon+R_\varepsilon
\right)
\right].
\label{eq:exact-ico-error}
\end{equation}
\end{proposition}

\begin{proof}
By \cref{prop:output-coherence-main}, it is enough to prove
\begin{equation}
\mu_{\misc,\gen}
\left(
\cN_{\rm AD}^{\varepsilon},
\cN_{\rm AD}^{\varepsilon}
\right)
=
s\left(
\varepsilon+R_\varepsilon
\right).
\label{eq:exact-ico-mu}
\end{equation}
For achievability, \cref{app:explicit-achievability} constructs a one-parameter family of general-process feasible points with $t\in[0,1]$, $T_1=0$, and diagonal population block. The family satisfies the general-process normalization and MISC constraints and is positive semidefinite. It achieves
\begin{equation}
x_c(t)
=2s\left[\varepsilon t+s\sqrt{t(1-t)}\right].
\label{eq:ico-family-xc}
\end{equation}
The parameter $t$ controls how the process distributes weight between the two order-sensitive sectors. The term $\varepsilon t$ in \cref{eq:ico-family-xc} is supported by relaxation-compatible histories, whereas $s\sqrt{t(1-t)}$ is an interference term between the two sectors. The square root is the characteristic coherence allowed between weights $t$ and $1-t$ by positivity.

Writing $u=2t-1$ turns the optimisation into the alignment of the unit vector $(u,\sqrt{1-u^2})$ with the channel vector $(\varepsilon,s)$. Cauchy--Schwarz gives
\begin{equation}
\varepsilon u+s\sqrt{1-u^2}
\le
\sqrt{\varepsilon^2+s^2}=R_\varepsilon.
\label{eq:ico-cauchy}
\end{equation}
Equality holds when the process weight is matched to the relative strengths of relaxation and coherent survival, at
\begin{equation}
t_\star
=
\frac12\left(1+\frac{\varepsilon}{R_\varepsilon}\right),
\label{eq:t-star}
\end{equation}
which gives $x_c(t_\star)=s(\varepsilon+R_\varepsilon)$.

For optimality, \cref{app:ico-certificate} constructs positive semidefinite operators $P_\varepsilon,Q_\varepsilon$ satisfying, for every reduced general-process feasible point with $T_1=0$,
\begin{equation}
s(\varepsilon+R_\varepsilon)-x_c
=
\langle P_\varepsilon,S+T_c\rangle
+
\langle Q_\varepsilon,S-T_c\rangle
\ge0.
\label{eq:ico-certificate-main}
\end{equation}
Thus no general-process strategy exceeds the value attained at $t_\star$.
\end{proof}

The closed form exposes a different balance in the general class. Relaxation enters through $\varepsilon$, coherent survival through $s$, and $R_\varepsilon=\sqrt{\varepsilon^2+s^2}$ combines the two. The optimal parameter $t_\star$ matches the process weights to this channel-dependent direction. Fixed-order comb normalization prevents the same simultaneous balance, which is exactly what the strict comparison below detects.

\subsection{Strict separation from fixed-order sequential strategies}

\begin{lemma}[Strict sequential-to-general-process separation]
\label{lem:strict-seq-ico-main}
For every $0<\varepsilon<1$,
\begin{equation}
D_{\misc}^{\seq}
\left(
\cN_{\rm AD}^{\varepsilon},
\cN_{\rm AD}^{\varepsilon};
\cI
\right)
>
D_{\misc}^{\gen}
\left(
\cN_{\rm AD}^{\varepsilon},
\cN_{\rm AD}^{\varepsilon};
\cI
\right).
\label{eq:strict-seq-ico-error}
\end{equation}
\end{lemma}

\begin{proof}
By \cref{prop:output-coherence-main}, it is enough to prove
\begin{equation}
\mu_{\misc,\seq}
\left(
\cN_{\rm AD}^{\varepsilon},
\cN_{\rm AD}^{\varepsilon}
\right)
<
\mu_{\misc,\gen}
\left(
\cN_{\rm AD}^{\varepsilon},
\cN_{\rm AD}^{\varepsilon}
\right).
\label{eq:strict-seq-ico-mu}
\end{equation}
Let
\begin{equation}
m
:=
\mu_{\misc,\gen}
\left(
\cN_{\rm AD}^{\varepsilon},
\cN_{\rm AD}^{\varepsilon}
\right)
=
s(\varepsilon+R_\varepsilon).
\end{equation}
Since $0<\varepsilon<1$, one has $s>0$ and therefore $m>0$. For every feasible $x,y,z$ in \cref{eq:sequential-variational}, direct expansion gives
\begin{equation}
\begin{aligned}
m-F_\varepsilon(x,y,z)
={}&
\left(
\sqrt m\,x
-
\frac{s}{\sqrt m}\sqrt{y^2+s^2z^2}
\right)^2\\
&+
\frac1m
\begin{pmatrix}y&z\end{pmatrix}
M_\varepsilon
\begin{pmatrix}y\\z\end{pmatrix},
\end{aligned}
\label{eq:strict-decomposition}
\end{equation}
where
\begin{equation}
M_\varepsilon
=
\begin{pmatrix}
m^2-s^2&-m\varepsilon s\\
-m\varepsilon s&m^2-s^4-m\varepsilon s
\end{pmatrix}.
\label{eq:M-epsilon}
\end{equation}
Its leading principal minor is
\begin{equation}
(M_\varepsilon)_{11}
=
\varepsilon(1-\varepsilon)
(2\varepsilon+2R_\varepsilon-1)>0,
\label{eq:M-leading-minor}
\end{equation}
and
\begin{equation}
\det M_\varepsilon
=
\varepsilon^2(1-\varepsilon)^2
\left[
(2\varepsilon-1)R_\varepsilon
+2\varepsilon^2-2\varepsilon+1
\right].
\label{eq:M-determinant}
\end{equation}
The bracket is positive for $\varepsilon\ge1/2$. For $0<\varepsilon<1/2$,
\begin{equation}
(2\varepsilon^2-2\varepsilon+1)^2
-(1-2\varepsilon)^2R_\varepsilon^2
=
\varepsilon(1-\varepsilon)>0,
\label{eq:M-det-small-epsilon}
\end{equation}
so the bracket remains positive. Sylvester's criterion gives $M_\varepsilon\succ0$. If $(y,z)\ne(0,0)$, the quadratic term in \cref{eq:strict-decomposition} is strictly positive. If $y=z=0$, then $x=1$ and the squared term is $m>0$. Hence $F_\varepsilon(x,y,z)<m$ at every feasible point. Compactness and \cref{prop:sequential-main} prove the claim.
\end{proof}

Equation~\eqref{eq:strict-decomposition} separates the gap into a squared mismatch and a positive-definite residual form. The square tests whether the direct coherent branch can be perfectly matched to the memory-assisted branch. Even if that matching condition is imposed, $M_\varepsilon\succ0$ leaves a strictly positive penalty for every nonzero sequential routing vector $(y,z)$. Hence the general-process gain cannot be removed by reoptimising either fixed-order memory. It is tied to correlations that cannot be reproduced by a definite-order sequential comb.

\begin{theorem}[Strict MISC causal hierarchy]
\label{thm:misc-hierarchy-main}
For every $0<\varepsilon<1$,
\begin{equation}
D_{\misc}^{\para}
\left(\cN_{\rm AD}^{\varepsilon},\cN_{\rm AD}^{\varepsilon};\cI\right)
>
D_{\misc}^{\seq}
\left(\cN_{\rm AD}^{\varepsilon},\cN_{\rm AD}^{\varepsilon};\cI\right)
>
D_{\misc}^{\gen}
\left(\cN_{\rm AD}^{\varepsilon},\cN_{\rm AD}^{\varepsilon};\cI\right).
\label{eq:misc-hierarchy-main}
\end{equation}
The same sequential value is obtained for the two fixed orders $A\prec B$ and $B\prec A$.
\end{theorem}

\begin{proof}
The two strict inequalities follow directly from
\cref{lem:no-parallel-saturation-main,lem:strict-seq-ico-main}. The two fixed-order sequential optima coincide by exchange symmetry of the identical resource channels.
\end{proof}

\section{Free program-state parallelization and mixed-Pauli collapse}
\label{sec:collapse}

The amplitude-damping hierarchy depends on channel structure; it disappears when the resource calls can be compiled into program states before the higher-order processing begins. The relevant sufficient condition is a channel-stretching statement at the level of free superchannels. One use of each resource channel prepares a channel-dependent state, a fixed CPTP interface reconstructs the channel from that state, and the states are all generated in parallel. What remains is a single CPTP post-processing map acting jointly on the stored program states and the eventual input.

\subsection{Common program-state simulation and parallelization}

\begin{definition}[Common program-state simulation]
\label{def:common-state-simulation}
For each input position $k=1,\ldots,n$, let $R_k$ be a finite-dimensional register, let
\begin{equation}
    \mathfrak P_k:
    \cE_k
    \longmapsto
    \omega_{\cE_k}\in\mathsf D(R_k)
    \label{eq:state-preparation-supermap}
\end{equation}
be a deterministic channel-to-state supermap using one occurrence of $\cE_k$, and let
\begin{equation}
    \cT_k:
    X_{k,i}R_k
    \longrightarrow
    X_{k,o}
    \label{eq:fixed-simulation-interface}
\end{equation}
be a CPTP map that may depend on $k$ but not on the particular resource channel being simulated. These objects define a one-input simulator
$\Phi_k:X_k\to X_k$ by
\begin{equation}
    \Phi_k(\cE_k)(\rho)
    :=
    \cT_k
    \left(
        \rho\ox\omega_{\cE_k}
    \right).
    \label{eq:state-simulation-interface}
\end{equation}

A resource tuple
\begin{equation}
    \vect{\cN}
    :=
    \left(
        \cN_1,\ldots,\cN_n
    \right)
\end{equation}
is \emph{commonly program-state simulable} when the preparation supermaps $\mathfrak P_1,\ldots,\mathfrak P_n$ can be executed in one parallel
layer and
\begin{equation}
    \Phi_k(\cN_k)=\cN_k,
    \qquad
    k=1,\ldots,n.
    \label{eq:simulation-fixed-points}
\end{equation}
\end{definition}

\begin{theorem}[Parallelization from a common program-state simulation]
\label{thm:common-simulation}
Let $\cS$ be an $n$-input superchannel, and suppose that the resource tuple $\vect{\cN}=(\cN_1,\ldots,\cN_n)$ satisfies Definition~\ref{def:common-state-simulation}. Then there exists a parallel superchannel $\cP_{\cS}$ satisfying
\begin{equation}
    \cP_{\cS}
    \left(
        \cE_1,\ldots,\cE_n
    \right)
    =
    \cS
    \left(
        \Phi_1(\cE_1),\ldots,\Phi_n(\cE_n)
    \right)
    \label{eq:parallelised-superchannel}
\end{equation}
for every tuple of input channels. In particular,
\begin{equation}
    \cP_{\cS}
    \left(
        \cN_1,\ldots,\cN_n
    \right)
    =
    \cS
    \left(
        \cN_1,\ldots,\cN_n
    \right).
    \label{eq:parallelised-same-output}
\end{equation}

Let $\mathbf F\in\{\misc,\disc\}$. If $\cS\in\mathbf F$ and every simulator $\Phi_k\in\mathbf F$, with the free classes understood in the sense of Definition~\ref{def:multi-input-misc-disc}, then $\cP_{\cS}\in\mathbf F$. Consequently, for every target channel $\cM$ and every causal class $\mathbf H$ satisfying
\begin{equation}
    \para\subseteq\mathbf H\subseteq\gen,
\end{equation}
one has
\begin{equation}
\begin{aligned}
    D_{\mathbf F}^{\para}
    \left(
        \vect{\cN};\cM
    \right)
    =
    D_{\mathbf F}^{\mathbf H}
    \left(
        \vect{\cN};\cM
    \right)
    =
    D_{\mathbf F}^{\gen}
    \left(
        \vect{\cN};\cM
    \right).
\end{aligned}
\label{eq:common-simulation-collapse}
\end{equation}
\end{theorem}

\begin{proof}
For an arbitrary density operator $\tau_k\in\mathsf D(R_k)$, define the channel
\begin{equation}
    \Phi_k^{\tau_k}(\rho)
    :=
    \cT_k(\rho\ox\tau_k).
    \label{eq:state-parametrised-channel}
\end{equation}
Let $J_{\cT_k}$ be the Choi operator of $\cT_k$. Linking the fixed interfaces into the Choi operator $\Theta^{\cS}$ gives the residual operator
\begin{equation}
    \Xi_{\cS}
    :=
    \left(
        \bigotimes_{k=1}^{n}J_{\cT_k}
    \right)
    *
    \Theta^{\cS},
    \label{eq:residual-choi}
\end{equation}
where the systems corresponding to the $n$ channel interfaces are contracted. Associativity of the link product gives
\begin{equation}
    J_{
        \cS(
            \Phi_1^{\tau_1},
            \ldots,
            \Phi_n^{\tau_n}
        )
    }
    =
    \left(
        \bigotimes_{k=1}^{n}\tau_k
    \right)
    *
    \Xi_{\cS}.
    \label{eq:residual-action}
\end{equation}

Complete positivity of the superchannel $\cS$, together with complete positivity of the interfaces $\cT_k$, implies
\begin{equation}
    \Xi_{\cS}\succeq0.
    \label{eq:residual-positive}
\end{equation}
Positivity alone is not enough; the residual map must also be trace preserving. Since every
$\Phi_k^{\tau_k}$ is a channel whenever $\tau_k$ is a density operator,
and $\cS$ maps tuples of channels to an output channel, one has
\begin{equation}
    \tr_{C_o}
    \left[
        \left(
            \bigotimes_{k=1}^{n}\tau_k
        \right)
        *
        \Xi_{\cS}
    \right]
    =
    I_{C_i}
    \label{eq:residual-normalisation-products}
\end{equation}
for every product of density operators.

Product density operators linearly span the Hermitian operator space on $R_1\cdots R_n$. Hence every Hermitian operator $X$ on these registers can be written as
\begin{equation}
    X
    =
    \sum_j
    c_j
    \tau_{1,j}\ox\cdots\ox\tau_{n,j},
    \label{eq:density-product-spanning}
\end{equation}
where each $\tau_{k,j}$ is a density operator. Since every product appearing in Eq.~\eqref{eq:density-product-spanning} has unit trace,
\begin{equation}
    \sum_j c_j=\tr X.
\end{equation}
By linearity and Eq.~\eqref{eq:residual-normalisation-products},
\begin{equation}
    \tr_{C_o}
    \left[
        X*\Xi_{\cS}
    \right]
    =
    (\tr X)I_{C_i}
    \label{eq:residual-normalisation-full-space}
\end{equation}
for every Hermitian $X$, and therefore for every operator by complex linearity. This is equivalent to the trace-preservation Choi condition for a channel
\begin{equation}
    \mathcal R_{\cS}:
    C_iR_1\cdots R_n
    \longrightarrow
    C_o
    \label{eq:residual-channel}
\end{equation}
whose Choi operator is $\Xi_{\cS}$. Thus
$\mathcal R_{\cS}$ is CPTP.

This residual channel gives the promised parallel implementation. The superchannel $\cP_{\cS}$ applies all preparation supermaps $\mathfrak P_k$ in parallel, producing
\begin{equation}
\omega_{\cE_1}\ox\cdots\ox\omega_{\cE_n}.
\end{equation}
It then applies $\mathcal R_{\cS}$ jointly to these states and the input on $C_i$. Since all resource-channel uses occur in the initial parallel layer, $\cP_{\cS}$ is a parallel superchannel. Equation \eqref{eq:residual-action} gives Eq.~\eqref{eq:parallelised-superchannel}, while the fixed-point identities in Eq.~\eqref{eq:simulation-fixed-points} give Eq.~\eqref{eq:parallelised-same-output}.

Moreover,
\begin{equation}
    \cP_{\cS}
    =
    \cS
    \circ
    \left(
        \Phi_1,\ldots,\Phi_n
    \right)
    \label{eq:parallel-map-as-composition}
\end{equation}
as an equality of higher-order maps on the full linear Choi space. If $\cS$ and all $\Phi_k$ are MISC, then Lemma~\ref{lem:free-composition-closure} implies that $\cP_{\cS}$ is MISC. The same lemma proves the DISC statement when $\cS$ and all $\Phi_k$ are DISC.

Thus every output reached by a general-process map in the free class $\mathbf F$ is also reached by a parallel map in that class, so
\begin{equation}
    D_{\mathbf F}^{\para}
    \left(
        \vect{\cN};\cM
    \right)
    \le
    D_{\mathbf F}^{\gen}
    \left(
        \vect{\cN};\cM
    \right).
    \label{eq:collapse-reverse-inequality}
\end{equation}
The causal inclusions imply
\begin{equation}
    D_{\mathbf F}^{\para}
    \left(
        \vect{\cN};\cM
    \right)
    \ge
    D_{\mathbf F}^{\mathbf H}
    \left(
        \vect{\cN};\cM
    \right)
    \ge
    D_{\mathbf F}^{\gen}
    \left(
        \vect{\cN};\cM
    \right).
    \label{eq:collapse-causal-inequalities}
\end{equation}
Combining these inequalities proves Eq.~\eqref{eq:common-simulation-collapse}.
\end{proof}

\subsection{Hierarchy collapse for mixed-Pauli channels}

Mixed-Pauli channels provide an explicit family for which the abstract program-state condition can be checked exactly. A qubit mixed-Pauli channel is
\begin{equation}
    \cN_{\vect p}(\rho)
    =
    \sum_{\alpha=0}^{3}
    p_\alpha
    \sigma_\alpha\rho\sigma_\alpha^\dagger,
    \qquad
    p_\alpha\ge0,
    \qquad
    \sum_{\alpha=0}^{3}p_\alpha=1,
    \label{eq:mixed-pauli}
\end{equation}
where
\begin{equation}
    \sigma_0=I,
    \qquad
    \sigma_1=X,
    \qquad
    \sigma_2=Y,
    \qquad
    \sigma_3=Z.
\end{equation}
Let
\begin{equation}
    \ket{\Phi_\alpha}
    :=
    (I\ox\sigma_\alpha)\ket{\Phi^+},
    \qquad
    \Phi_\alpha
    :=
    \ketbra{\Phi_\alpha}{\Phi_\alpha},
    \qquad
    \ket{\Phi^+}
    :=
    \frac{\ket{00}+\ket{11}}{\sqrt2}.
    \label{eq:bell-projectors}
\end{equation}
The normalised Choi state of $\cN_{\vect p}$ is
\begin{equation}
    \omega_{\cN_{\vect p}}
    :=
    \frac12J_{\cN_{\vect p}}
    =
    \sum_{\alpha=0}^{3}
    p_\alpha\Phi_\alpha.
    \label{eq:mixed-pauli-bell-choi}
\end{equation}
Standard teleportation through $\omega_{\cN_{\vect p}}$ implements $\cN_{\vect p}$ using a fixed Bell measurement and fixed Pauli corrections that do not depend on $\vect p$ \cite{bennett1993teleporting,pirandola2017fundamental}.

For an arbitrary qubit channel $\cE$, define the teleportation-induced simulator
\begin{equation}
    \mathfrak T_{\rm tel}(\cE)(\rho)
    :=
    \cT_{\rm tel}
    \left(
        \rho\ox\frac12J_{\cE}
    \right),
    \label{eq:teleportation-induced}
\end{equation}
where $\cT_{\rm tel}$ is the fixed Bell-measurement and Pauli-correction channel.

\begin{lemma}[Teleportation simulation is DISC]
\label{lem:teleportation-free-properties}
The teleportation simulator $\mathfrak T_{\rm tel}$ satisfies the DISC condition on the full linear qubit Choi space and is therefore also MISC. Moreover, every mixed-Pauli channel is a fixed point,
\begin{equation}
    \mathfrak T_{\rm tel}
    \left(
        \cN_{\vect p}
    \right)
    =
    \cN_{\vect p}.
    \label{eq:mixed-pauli-teleportation-fixed-point}
\end{equation}
\end{lemma}

\begin{proof}
Let $X\in\mathsf L(X_iX_o)$ be an arbitrary operator, not necessarily the Choi operator of a channel. The linear Choi-space action induced by teleportation is
\begin{equation}
    \widehat{\mathfrak T}_{\rm tel}(X)
    =
    \sum_{\alpha=0}^{3}
    \tr(\Phi_\alpha X)\Phi_\alpha.
    \label{eq:teleportation-full-linear-action}
\end{equation}
Indeed, for $X=J_{\cE}$, the corresponding Pauli probability is
\begin{equation}
    q_\alpha(\cE)
    =
    \tr
    \left(
        \Phi_\alpha\frac{J_{\cE}}2
    \right),
\end{equation}
and the Choi operator of the resulting Pauli channel is
\begin{equation}
    2\sum_{\alpha=0}^{3}q_\alpha(\cE)\Phi_\alpha,
\end{equation}
which is Eq.~\eqref{eq:teleportation-full-linear-action}.

Complete dephasing of the Bell projectors gives
\begin{equation}
\begin{aligned}
    \cD(\Phi_0)
    &=
    \cD(\Phi_3)
    =
    \frac{\Phi_0+\Phi_3}{2},
    \\
    \cD(\Phi_1)
    &=
    \cD(\Phi_2)
    =
    \frac{\Phi_1+\Phi_2}{2}.
\end{aligned}
\label{eq:dephased-bell-projectors}
\end{equation}
Define
\begin{equation}
    a_\alpha(X):=\tr(\Phi_\alpha X).
\end{equation}
Equation~\eqref{eq:dephased-bell-projectors} implies
\begin{equation}
\begin{aligned}
    \cD
    \left[
        \widehat{\mathfrak T}_{\rm tel}(X)
    \right]
    ={}&
    \frac{a_0(X)+a_3(X)}2
    (\Phi_0+\Phi_3)
    \\
    &+
    \frac{a_1(X)+a_2(X)}2
    (\Phi_1+\Phi_2).
\end{aligned}
\label{eq:dephasing-after-teleportation}
\end{equation}
Since $\cD$ is self-adjoint,
\begin{equation}
\begin{aligned}
    a_0(\cD(X))
    &=
    a_3(\cD(X))
    =
    \frac{a_0(X)+a_3(X)}2,
    \\
    a_1(\cD(X))
    &=
    a_2(\cD(X))
    =
    \frac{a_1(X)+a_2(X)}2.
\end{aligned}
\label{eq:bell-weights-after-input-dephasing}
\end{equation}
Substituting Eq.~\eqref{eq:bell-weights-after-input-dephasing} into Eq.~\eqref{eq:teleportation-full-linear-action} gives
\begin{equation}
    \cD
    \circ
    \widehat{\mathfrak T}_{\rm tel}
    =
    \widehat{\mathfrak T}_{\rm tel}
    \circ
    \cD
    \label{eq:teleportation-disc-full-space}
\end{equation}
on the full linear Choi space. By Proposition~\ref{prop:multi-input-choi-characterisation}, this is exactly the DISC condition for the one-input superchannel $\mathfrak T_{\rm tel}$. Hence
\begin{equation}
    \mathfrak T_{\rm tel}\in\disc\subseteq\misc.
\end{equation}

Finally, using Eq.~\eqref{eq:mixed-pauli-bell-choi},
\begin{equation}
\begin{aligned}
    \widehat{\mathfrak T}_{\rm tel}
    \left(
        J_{\cN_{\vect p}}
    \right)
    &=
    \sum_{\alpha=0}^{3}
    \tr
    \left(
        \Phi_\alpha
        2\sum_{\beta=0}^{3}p_\beta\Phi_\beta
    \right)
    \Phi_\alpha
    \\
    &=
    2\sum_{\alpha=0}^{3}p_\alpha\Phi_\alpha
    =
    J_{\cN_{\vect p}}.
\end{aligned}
\end{equation}
This proves Eq.~\eqref{eq:mixed-pauli-teleportation-fixed-point}.
\end{proof}

\begin{theorem}[Mixed-Pauli collapse under MISC and DISC]
\label{thm:mixed-pauli-collapse-main}
Let
\begin{equation}
    \vect{\cN}
    :=
    \left(
        \cN_{\vect p_1},
        \ldots,
        \cN_{\vect p_n}
    \right)
\end{equation}
be any finite collection of qubit mixed-Pauli channels, not necessarily identical. For every permutation $\pi\in S_n$, let $\seq_\pi$ denote the fixed-order sequential class in which the resource channels are used according to
\begin{equation}
    \pi(1)\prec\pi(2)\prec\cdots\prec\pi(n).
\end{equation}
Then, for every target channel $\cM$, every
$\mathbf F\in\{\misc,\disc\}$, and every $\pi\in S_n$,
\begin{equation}
\begin{aligned}
    D_{\mathbf F}^{\para}
    \left(
        \vect{\cN};\cM
    \right)
    =
    D_{\mathbf F}^{\seq_\pi}
    \left(
        \vect{\cN};\cM
    \right)
    =
    D_{\mathbf F}^{\gen}
    \left(
        \vect{\cN};\cM
    \right).
\end{aligned}
\label{eq:mixed-pauli-collapse-both}
\end{equation}
Thus, for any finite collection of qubit mixed-Pauli resource channels, neither fixed-order sequential nor general-process strategies achieve a smaller optimal transformation error than parallel processing.
\end{theorem}

\begin{proof}
For every $k$, the normalised Choi state
\begin{equation}
    \omega_{\cE_k}
    =
    \frac12J_{\cE_k}
\end{equation} 
can be prepared by applying $\cE_k$ to one half of a maximally entangled state. These preparations use each resource channel exactly once and can be executed in one parallel layer. The resulting states are stored and supplied to the same fixed teleportation interface $\mathfrak T_{\rm tel}$ whenever the original higher-order protocol would use the corresponding resource channel. This explicit teleportation-stretching construction is described in \cref{app:teleportation-details}.

By Lemma~\ref{lem:teleportation-free-properties}, the associated simulator $\mathfrak T_{\rm tel}$ is DISC and hence MISC, and every supplied mixed-Pauli channel is one of its fixed points. The hypotheses of Theorem~\ref{thm:common-simulation} therefore hold under both free constraints.

Let $\cS\in\gen\cap\mathbf F$ be arbitrary. Theorem \ref{thm:common-simulation} produces a parallel superchannel $\cP_{\cS}\in\para\cap\mathbf F$ with the same output on $\vect{\cN}$. Hence
\begin{equation}
    D_{\mathbf F}^{\para}
    \left(
        \vect{\cN};\cM
    \right)
    \le
    D_{\mathbf F}^{\gen}
    \left(
        \vect{\cN};\cM
    \right).
\end{equation}
For every $\pi\in S_n$,
\begin{equation}
    \para
    \subseteq
    \seq_\pi
    \subseteq
    \gen,
\end{equation}
and therefore
\begin{equation}
    D_{\mathbf F}^{\para}
    \left(
        \vect{\cN};\cM
    \right)
    \ge
    D_{\mathbf F}^{\seq_\pi}
    \left(
        \vect{\cN};\cM
    \right)
    \ge
    D_{\mathbf F}^{\gen}
    \left(
        \vect{\cN};\cM
    \right).
\end{equation}
Combining the two directions proves
Eq.~\eqref{eq:mixed-pauli-collapse-both}.
\end{proof}

\section{Robustness under DISC}
\label{sec:disc-robustness}

The DISC comparison tests whether the amplitude-damping hierarchy survives a genuinely smaller free set. In the reduced variables, MISC constrains only the coherence-carrying block $T_c$, whereas DISC also dephases the population blocks $T_0$ and $T_1$. Those extra conditions can change an optimum because they forbid coherence from being stored in population-processing sectors. Since $\disc\subseteq\misc$,
\begin{equation}
D_{\misc}^{\mathbf H}
\le
D_{\disc}^{\mathbf H}
\label{eq:misc-disc-general-order}
\end{equation}
for every task and causal class, with strict inequality possible in general. For the two channels defined exactly in \cref{app:misc-disc-example}, the corresponding identity-target errors are reported in \cref{tab:misc-disc-numerical}. The finite-decimal Choi matrices in that appendix are the exact definition of the revised numerical instance rather than rounded displays of a separate full-precision input file.
\begin{table}[h!]
\centering
\caption{Optimal transformation errors under MISC and DISC for the exact channel pair defined in \cref{app:misc-disc-example}, using parallel processing, the fixed order $A\prec B$, and the general-process class. The displayed values are numerical SDP optima rounded to seven decimal places.}
\label{tab:misc-disc-numerical}
\begin{tabular}{c@{\qquad}cc}
\toprule
$\mathbf H$
&
$D_{\misc}^{\mathbf H}$
&
$D_{\disc}^{\mathbf H}$
\\
\midrule
$\para$ & 0.1960802 & 0.2064456 \\
$A\prec B$ & 0.1956927 & 0.2064212 \\
$\gen$  & 0.1912633 & 0.2047563 \\
\bottomrule
\end{tabular}
\end{table}

For this exact pair, the six displayed values reproduce the numerical separation between MISC and DISC, with each MISC--DISC gap exceeding $10^{-2}$. These decimals are solver outputs rather than exact constants. DISC is therefore not redundant in the present framework. Amplitude damping is more special: its optimum already lies in the DISC-feasible subset.

\begin{proposition}[MISC--DISC equality for amplitude damping]
\label{prop:misc-disc-equality-main}
For every $0\le\varepsilon\le1$ and every
$\mathbf H\in\{\para,\seq,\gen\}$,
\begin{equation}
D_{\misc}^{\mathbf H}
\left(\cN_{\rm AD}^{\varepsilon},\cN_{\rm AD}^{\varepsilon};\cI\right)
=
D_{\disc}^{\mathbf H}
\left(\cN_{\rm AD}^{\varepsilon},\cN_{\rm AD}^{\varepsilon};\cI\right).
\label{eq:misc-disc-equality-main}
\end{equation}
\end{proposition}

\begin{proof}
It remains only to prove $D_{\disc}^{\mathbf H}\le D_{\misc}^{\mathbf H}$. For either fixed sequential order, the construction in \cref{eq:seq-achievability-S,eq:seq-achievability-Tc} attains the MISC optimum with $T_1=0$, $T_0=S$ diagonal, and $\cD_{AB}(T_c)=0$. It therefore satisfies the DISC constraints and achieves the same optimum under DISC. The opposite fixed order has the same optimum by exchanging the two identical resource channels.

For the general-process class, the optimal family \cref{eq:ico-family-S,eq:ico-family-Tc} has the same properties and is DISC feasible at $t=t_\star$. Hence it attains the MISC optimum under DISC.

For the parallel class, \cref{prop:output-coherence-main} permits $T_1=0$. Parallel normalisation then gives $T_0=S=\rho_{A_iB_i}\ox I_{A_oB_o}$ for a density operator $\rho_{A_iB_i}$. The amplitude-damping Choi operator is invariant under independent joint input--output phase rotations on the two channel uses. Averaging the feasible point over these symmetries preserves $x_c$, positivity, and parallel normalisation while replacing $\rho_{A_iB_i}$ by its complete dephasing. The averaged point satisfies $T_0=\cD_{AB}(T_0)$ and $\cD_{AB}(T_c)=0$, and is therefore DISC feasible. The explicit covariance calculation is given in \cref{app:parallel-disc-twirl}.

At $\varepsilon=0$, one has $\cN_{\rm AD}^{0}=\cI$. Consider the parallel superchannel that wires $C_i$ directly to $A_i$, wires $A_o$ to $C_o$, feeds an incoherent fixed state into $B_i$, and discards $B_o$. This construction uses both supplied interfaces in one layer, returns the first resource channel exactly, maps classical channels to classical channels, and commutes with dephasing. It is therefore feasible under MISC and DISC and belongs to every causal class. Consequently, for every $\mathbf F\in\{\misc,\disc\}$ and $\mathbf H\in\{\para,\seq,\gen\}$,
\begin{equation}
D_{\mathbf F}^{\mathbf H}
\left(
    \cN_{\rm AD}^{0},\cN_{\rm AD}^{0};
    \cI
\right)
=
0.
\end{equation}

At $\varepsilon=1$, the amplitude-damping channel becomes the reset channel $\cN_{\rm AD}^{1}(\rho)=\ketbra{0}{0}$, whose Choi operator is diagonal in the incoherent basis. Thus $\cN_{\rm AD}^{1}$ is a classical channel. Because both MISC and DISC superchannels map classical resource channels to a classical output channel, the output coherence coefficient necessarily satisfies $x_c=0$. Proposition~\ref{prop:output-coherence-main} therefore implies
\begin{equation}
D_{\mathbf F}^{\mathbf H}
\left(
    \cN_{\rm AD}^{1},\cN_{\rm AD}^{1};
    \cI
\right)
\ge
\frac12.
\end{equation}

Conversely, consider the parallel replacement superchannel that ignores the two resource channels and outputs the completely dephasing channel $\cD$. This superchannel is feasible under both MISC and DISC and is contained in every causal class. Since $J_{\cD} = \Omega_0$  and  $\frac12 \left\|\cD-\cI \right\|_{\diamond} = \frac12$,
the lower bound is attainable. Consequently, for every $\mathbf F\in\{\misc,\disc\}$ and $\mathbf H\in\{\para,\seq,\gen\}$,
\begin{equation}
D_{\mathbf F}^{\mathbf H}
\left(
    \cN_{\rm AD}^{1},\cN_{\rm AD}^{1};
    \cI
\right)
=
\frac12.
\end{equation}
\end{proof}

\begin{theorem}[Strict DISC causal hierarchy]
\label{thm:disc-hierarchy-main}
For every $0<\varepsilon<1$,
\begin{equation}
D_{\disc}^{\para}
\left(
\cN_{\rm AD}^{\varepsilon},
\cN_{\rm AD}^{\varepsilon};
\cI
\right)
>
D_{\disc}^{\seq}
\left(
\cN_{\rm AD}^{\varepsilon},
\cN_{\rm AD}^{\varepsilon};
\cI
\right)
>
D_{\disc}^{\gen}
\left(
\cN_{\rm AD}^{\varepsilon},
\cN_{\rm AD}^{\varepsilon};
\cI
\right).
\label{eq:disc-hierarchy-main}
\end{equation}
Moreover,
\begin{equation}
D_{\disc}^{\gen}
\left(
\cN_{\rm AD}^{\varepsilon},
\cN_{\rm AD}^{\varepsilon};
\cI
\right)
=
\frac12
\left[
1-
s\left(
\varepsilon+R_\varepsilon
\right)
\right].
\label{eq:disc-ico-exact}
\end{equation}
\end{theorem}

\begin{proof}
Combine \cref{thm:misc-hierarchy-main,prop:misc-disc-equality-main,prop:ico-main}.
\end{proof}

Why does DISC leave the amplitude-damping optimum unchanged? The optimal constructions already take $T_0$ and $T_1$ diagonal, so the extra DISC constraints remove no useful degree of freedom. Those blocks perform only routing and normalization; the phase-sensitive part sits entirely in $T_c$, where MISC and DISC impose the same condition. The causal hierarchy therefore survives the stronger coherence restriction for structural, rather than numerical, reasons.

\section{Numerical results beyond the analytical example}
\label{sec:numerics}

The numerical calculations serve two purposes: they check the analytic amplitude-damping formulas and probe whether the same ordering appears away from that family. All primal and dual SDPs were solved with \textit{SeDuMi}~\cite{Sturm1999SeDuMi}. \Cref{fig:combined-results}(a) shows the amplitude-damping errors. Primal and dual values agree within $10^{-8}$ for every causal class and reproduce the strict ordering in \cref{thm:misc-hierarchy-main,thm:disc-hierarchy-main}. Because the MISC and DISC optima coincide analytically, the same three curves represent both free-operation classes.

We then sampled $1000$ independent pairs of random qubit channels, again with the identity target, to see how often a fixed-order-to-general-process gap appears outside the analytic example. For each instance, we computed
\begin{equation}
\Delta_{\seq-\gen}
:=
D_{\misc}^{\seq_{\fixAB}}(\cN_A,\cN_B;\cI)
-
D_{\misc}^{\gen}(\cN_A,\cN_B;\cI).
\label{eq:random-gap}
\end{equation}
To certify positivity using both programs, define
\begin{equation}
\Delta_{\rm cert}
:=
D_{\misc,\mathrm{dual}}^{\seq_{\fixAB}}
-
D_{\misc,\mathrm{primal}}^{\gen}.
\label{eq:certified-gap}
\end{equation}
A value $\Delta_{\rm cert}>0$ provides a solver-tolerance-certified lower bound on the fixed-order-to-general-process gap. In the supplied data, $989$ instances have a gap above $10^{-5}$, while $11$ lie below the adopted resolution and are treated as unresolved rather than as zero-gap instances. Across all programs, primal and dual values agree within $10^{-7}$.

\begin{figure}[h!]
\centering
\IfFileExists{figures/Figure_combined.png}{%
  \includegraphics[width=0.98\linewidth]{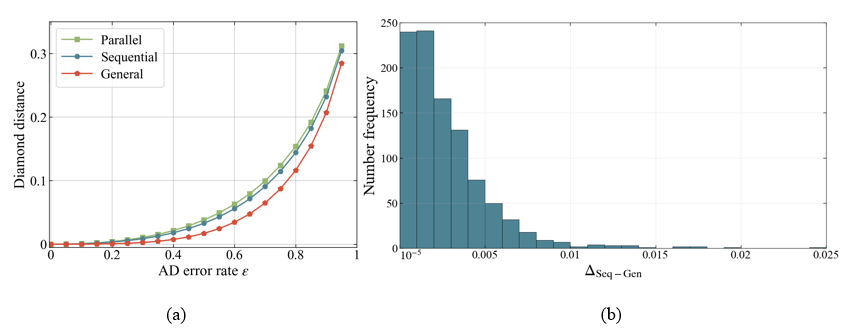}%
}{%
  \fbox{\parbox[c][6.0cm][c]{0.9\linewidth}{\centering
  Insert \texttt{figures/Figure\_combined.png}: (a) amplitude-damping errors and (b) random-channel gap histogram.}}%
}
\caption{(a) Optimal transformation errors for $(\cN_{\rm AD}^{\varepsilon})^{\otimes2}\to\cI$. The displayed curves apply to both MISC and DISC. The general-process class gives the smallest error, while parallel processing gives the largest. (b) Distribution of the fixed-order $A\prec B$-to-general-process improvement for $1000$ random pairs of qubit channels under MISC. The histogram begins at $10^{-5}$; the $11$ instances below this numerical resolution are reported separately as unresolved.}
\label{fig:combined-results}
\end{figure}

\begin{table}[h!]
\centering
\caption{Summary of the numerical evidence beyond the analytic amplitude-damping theorem. The random-channel counts refer to the reported sample and are not used as a distributional or genericity claim.}
\label{tab:numerical-summary}
\begin{tabular}{lc}
\toprule
Quantity & Reported value \\
\midrule
Random qubit-channel pairs & $1000$ \\
Certified gap above $10^{-5}$ & $989$ \\
Below numerical resolution & $11$ \\
Largest reported primal--dual mismatch & $10^{-7}$ \\
\bottomrule
\end{tabular}
\end{table}

These data are descriptive rather than distributional: they do not establish a generic advantage. They do show that the amplitude-damping separation is not isolated within the reported sample. The natural classification question is therefore which channel features block a common free simulation and leave room for a general process to outperform fixed-order memory.

\section{Discussion}
\label{sec:discussion}
The main conclusion is that causal organization and dynamical-resource restrictions cannot, in general, be optimized independently. For two uses of an amplitude-damping channel, identity-target symmetry turns the problem into the transfer of a single coherence coefficient. Within that reduced problem, explicit feasible constructions and analytic converse certificates separate parallel, fixed-order, and general processes for every nontrivial damping strength. The equality of the MISC and DISC optima further shows that this hierarchy survives a stricter free-operation constraint.

Mixed-Pauli channels expose the opposite regime and clarify why the causal hierarchy can disappear. Teleportation simulation transfers the channel dependence to Bell-diagonal program states, all of which can be prepared before any later processing. Because the induced simulator is DISC, the compilation remains free under both resource theories considered here. Once the channel calls have been replaced by stored program states, their temporal placement ceases to affect the attainable output. This identifies free program-state parallelization as a concrete obstruction to causal enhancement.

The scope of the amplitude-damping theorem is deliberately narrower than a device-level indefinite-order claim. The class $\gen$ is the general deterministic higher-order cone defined by the Choi normalization constraints used in the SDP. We have not classified its optimizer as a quantum switch or another quantum-controlled-order process, and we have not optimized over causal architectures beyond the parallel and two fixed-order classes treated explicitly. What is established is therefore a strict hierarchy among these specified process classes.

Two questions follow directly from the present results. One is structural: characterize channel families that admit a free program-state simulation, and contrast them with families whose coherence-carrying pathways remain constrained by fixed-order memory. The other is resource-theoretic: determine how the same causal comparison changes for dynamical resources such as entanglement, asymmetry, athermality, or magic. Both directions can be pursued within the same separation between causal constraints and free-operation constraints used here.

\section*{Data and code availability}

The exact channel pair used in \cref{tab:misc-disc-numerical} is defined in \cref{app:misc-disc-example} by Gaussian-integer Choi matrices with common denominator $10^4$. The reported values are numerical SDP optima for these matrices. Numerical code, solver outputs, verification scripts, random-channel data, and figure data are available at the \href{https://github.com/real-lin-zhu/Causal-Class-Hierarchies-in-Coherence-Constrained-Channel-Transformation.git}{publuc Github project repository}.

\section*{AI usage statement}
The authors formulated the research question and used several frontier LLMs under human direction to assist with exploratory derivations, proof organization, symbolic checks, numerical workflow design, and manuscript editing. The authors independently selected the final problem formulation, derived the analytical results, constructed and verified the semidefinite programs and analytic certificates, checked every displayed calculation, and reviewed the final manuscript. All analytical claims were independently verified by the authors, while numerical claims were checked through matched primal--dual optimization and explicit feasibility tests. The authors take full responsibility for the content and correctness of this article.

\section*{Acknowledgments}
This work was partially supported by the National Key R\&D Program of China (Grant No.~2024YFB4504004), the National Natural Science Foundation of China (Grant No.92576114), the Guangdong Provincial Quantum Science Strategic Initiative (Grant No.~GDZX2403008, GDZX2503001), and the Guangdong Provincial Key Lab of Integrated Communication, Sensing and Computation for Ubiquitous Internet of Things (Grant No.~2023B1212010007). X Zhao and B Zhao have been supported by the Hong Kong Research Grant Council through the Senior Research Fellowship SRFS2021-7S02.

\FloatBarrier

\ifprojectbib
\bibliographystyle{unsrt}
\bibliography{ref}
\else
\section*{References}
\emph{The project bibliography file \texttt{ref.bib} was not included in the uploaded source package used to compile this preview. Place the verified project bibliography beside this file to restore the full reference list and citation markers.}
\fi

\newpage
\appendix

\section{Notations and conventions}
\label{app:notation}

For a system $X$, $I_X$ is the identity operator, $\id_X$ is the identity superoperator on $\mathsf L(X)$, and $\tr_X$ is the partial trace. The Hilbert--Schmidt inner product is $\langle X,Y\rangle:=\tr(X^\dagger Y)$. For Hermitian operators, $X\succeq0$ means positive semidefinite.

For $M\in\mathsf L(XY)$ and $N\in\mathsf L(YZ)$, the link product over $Y$ is
\begin{equation}
M*_Y N
:=
\tr_Y\left[
(M^{T_Y}\ox I_Z)(I_X\ox N)
\right]\in\mathsf L(XZ),
\label{eq:link-product}
\end{equation}
where $T_Y$ denotes partial transpose in the fixed incoherent basis~\cite{Chiribella2009Networks}. The connected register is omitted when clear.

Let $R:=A_iA_oB_iB_o$. In the computational basis
$\ket{a_i,a_o,b_i,b_o}$, ordered lexicographically with index
$1+8a_i+4a_o+2b_i+b_o$, define $E_{i,j}:=\ket{i}\!\bra{j}$.

For compact causal constraints, define normalised replacement maps
\begin{equation}
\mathsf P_{X_o}(M)
:=
\tr_{X_o}M\ox\frac{I_{X_o}}{d_{X_o}},
\qquad
\mathsf P_{A_oB_o}(M)
:=
\tr_{A_oB_o}M\ox\frac{I_{A_oB_o}}{d_{A_oB_o}},
\label{eq:replacement-maps}
\end{equation}
with tensor factors reinserted in the fixed register order.

\section{General primal SDPs}
\label{app:general-sdp}

Applying the comb recursion in the register order $A_iA_oB_iB_oC_iC_o$ gives the causal normalization constraints below. The replacement maps in \eqref{eq:replacement-maps} keep them consistent with the unnormalized Choi convention \eqref{eq:choi-definition} and the link product \eqref{eq:link-product}~\cite{Chiribella2008CircuitArchitecture,Chiribella2009Networks,Bavaresco2021Hierarchy}.

For two input channels, the parallel constraints are
\begin{equation}
\begin{aligned}
\tr_{C_o}\Theta
&=
\tr_{A_oB_oC_o}\Theta\ox\frac{I_{A_oB_o}}{d_{A_oB_o}},\\
\tr_{C_oA_iB_i}\Theta
&=I_{A_oB_oC_i}.
\end{aligned}
\label{eq:parallel-full-constraints}
\end{equation}
For the fixed order $A\prec B$,
\begin{equation}
\begin{aligned}
\tr_{C_o}\Theta
&=
\tr_{B_oC_o}\Theta\ox\frac{I_{B_o}}{d_{B_o}},\\
\frac1{d_{B_o}}\tr_{C_oB_iB_o}\Theta
&=
\left(
\frac1{d_{B_o}}\tr_{A_oB_iB_oC_o}\Theta
\right)\ox\frac{I_{A_o}}{d_{A_o}},\\
\tr_{A_iA_oB_iB_oC_o}\Theta
&=d_{A_oB_o}I_{C_i}.
\end{aligned}
\label{eq:sequential-full-constraints}
\end{equation}
The opposite order follows by exchanging $A$ and $B$. In the analytical amplitude-damping example the two resources are identical, so the two fixed-order optima coincide by exchange symmetry.

For the full class $\gen$,
\begin{equation}
\begin{aligned}
\tr_{C_o}\Theta
={}&
\mathsf P_{A_o}(\tr_{C_o}\Theta)
+
\mathsf P_{B_o}(\tr_{C_o}\Theta)
-
\mathsf P_{A_oB_o}(\tr_{C_o}\Theta),\\
\tr_{A_iA_oC_o}\Theta
={}&
\tr_{A_iA_oB_oC_o}\Theta\ox\frac{I_{B_o}}{d_{B_o}},\\
\tr_{B_iB_oC_o}\Theta
={}&
\tr_{B_iB_oA_oC_o}\Theta\ox\frac{I_{A_o}}{d_{A_o}},\\
\tr_{A_iA_oB_iB_oC_o}\Theta
={}&d_{A_oB_o}I_{C_i}.
\end{aligned}
\label{eq:ico-full-constraints}
\end{equation}
Combining any of these constraints with positivity, the MISC or DISC equation, and the Watrous constraints in \cref{eq:general-sdp} gives the full primal SDP.

\section{Reduced MISC and DISC SDPs}
\label{app:reduced-sdps}

Set $S:=T_0+T_1$. The common positivity and scalar Watrous constraints are
\begin{equation}
\begin{aligned}
&T_1\succeq0,
\qquad
T_0\pm T_c\succeq0,\\
&z_1\ge0,
\qquad
z_0\pm z_c\ge0,
\qquad
z_0+z_1\le l,\\
&z_1\ge x_1,\\
&z_0+z_c-x_0-x_c+2\ge0,\\
&z_0-z_c-x_0+x_c\ge0.
\end{aligned}
\label{eq:common-reduced-constraints}
\end{equation}
MISC adds $\cD_{AB}(T_c)=0$. DISC adds
\begin{equation}
T_0=\cD_{AB}(T_0),
\qquad
T_1=\cD_{AB}(T_1),
\qquad
\cD_{AB}(T_c)=0.
\label{eq:reduced-disc-constraints}
\end{equation}

The reduced causal constraints are as follows. For $\para$,
\begin{equation}
S=\mathsf P_{A_oB_o}(S),
\qquad
\tr_{A_iB_i}S=I_{A_oB_o}.
\label{eq:parallel-reduced}
\end{equation}
For $A\prec B$,
\begin{equation}
S=\mathsf P_{B_o}(S),
\qquad
\tr_{B_iB_o}S
=
\mathsf P_{A_o}(\tr_{B_iB_o}S),
\qquad
\tr S=4.
\label{eq:sequential-reduced}
\end{equation}
For $\gen$,
\begin{equation}
\begin{aligned}
S={}&\mathsf P_{A_o}(S)+\mathsf P_{B_o}(S)-\mathsf P_{A_oB_o}(S),\\
\tr_{A_iA_o}S={}&\mathsf P_{B_o}(\tr_{A_iA_o}S),\\
\tr_{B_iB_o}S={}&\mathsf P_{A_o}(\tr_{B_iB_o}S),\\
\tr S={}&4.
\end{aligned}
\label{eq:ico-reduced}
\end{equation}
For identical resources, parallel and general-process variables may additionally be averaged under $A\leftrightarrow B$.

To write the duals compactly, set $U:=T_0+T_c$, $V:=T_0-T_c$, and $W:=T_1$, and let $\mathsf A_{\mathbf H}(S)=b_{\mathbf H}$ denote the corresponding causal equations. Define
\begin{equation}
\ell_{\cN}(T):=J_{\cN}^{\ox2}*T=\langle G_{\cN},T\rangle.
\label{eq:ell-definition}
\end{equation}
The reduced MISC dual is
\begin{equation}
\begin{aligned}
\text{maximise}\quad&
\langle b_{\mathbf H},Y_{\mathbf H}\rangle-2r_+\\
\text{subject to}\quad&
r_+G_{\cN}-\frac12\mathsf A_{\mathbf H}^*(Y_{\mathbf H})
-\frac12\cD_{AB}(H_c)\succeq0,\\
& r_-G_{\cN}-\frac12\mathsf A_{\mathbf H}^*(Y_{\mathbf H})
+\frac12\cD_{AB}(H_c)\succeq0,\\
& r_1G_{\cN}-\mathsf A_{\mathbf H}^*(Y_{\mathbf H})\succeq0,\\
&0\le r_1\le1,
\qquad
0\le r_\pm\le\frac12.
\end{aligned}
\label{eq:misc-reduced-dual}
\end{equation}
Let $\cQ_{AB}:=\id_{AB}-\cD_{AB}$. The DISC dual adds multipliers $H_0,H_1$ and replaces the first two matrix inequalities by
\begin{equation}
\begin{aligned}
r_+G_{\cN}
-\frac12\mathsf A_{\mathbf H}^*(Y_{\mathbf H})
-\frac12\cQ_{AB}(H_0)
-\frac12\cD_{AB}(H_c)&\succeq0,\\
r_-G_{\cN}
-\frac12\mathsf A_{\mathbf H}^*(Y_{\mathbf H})
-\frac12\cQ_{AB}(H_0)
+\frac12\cD_{AB}(H_c)&\succeq0,
\end{aligned}
\label{eq:disc-dual-first-two}
\end{equation}
and the third by
\begin{equation}
r_1G_{\cN}
-\mathsf A_{\mathbf H}^*(Y_{\mathbf H})
-\cQ_{AB}(H_1)\succeq0.
\label{eq:disc-dual-third}
\end{equation}
The class-dependent adjoints are
\begin{equation}
\mathsf A_{\para}^*(Y_0,Y_1)
=
Y_0-\mathsf P_{A_oB_o}(Y_0)+I_{A_iB_i}\ox Y_1,
\label{eq:parallel-adjoint}
\end{equation}
\begin{equation}
\begin{aligned}
\mathsf A_{\seq}^*(Y_0,Y_1,\eta)
={}&Y_0-\mathsf P_{B_o}(Y_0)\\
&+[Y_1-\mathsf P_{A_o}(Y_1)]\ox I_{B_iB_o}+\eta I_R,
\end{aligned}
\label{eq:sequential-adjoint}
\end{equation}
and
\begin{equation}
\begin{aligned}
\mathsf A_{\gen}^*(Y_0,Y_A,Y_B,\eta)
={}&Y_0-\mathsf P_{A_o}(Y_0)-\mathsf P_{B_o}(Y_0)
+\mathsf P_{A_oB_o}(Y_0)\\
&+I_{A_iA_o}\ox[Y_A-\mathsf P_{B_o}(Y_A)]\\
&+[Y_B-\mathsf P_{A_o}(Y_B)]\ox I_{B_iB_o}+\eta I_R.
\end{aligned}
\label{eq:ico-adjoint}
\end{equation}

\section{Explicit amplitude-damping achievability constructions}
\label{app:explicit-achievability}

This appendix collects the matrix-valued feasible constructions used for attainment in the amplitude-damping theorems. Keeping them here leaves the main text focused on the operational errors and the scalar optimization.

\subsection{Fixed-order sequential construction}
For achievability, choose $x,y,z\ge0$ with
$x^2+y^2+z^2=1$ and define the diagonal comb
\begin{equation}
\begin{aligned}
S=\diag\bigg(
&x^2,x^2,y^2,y^2,
\frac{x^2+y^2}{2},\frac{x^2+y^2}{2},
\frac{x^2+y^2}{2},\frac{x^2+y^2}{2},\\
&\frac{z^2}{2},\frac{z^2}{2},
\frac{z^2}{2},\frac{z^2}{2},
0,0,z^2,z^2
\bigg).
\end{aligned}
\label{eq:seq-achievability-S}
\end{equation}
Let
\begin{equation}
w:=\sqrt{y^2+s^2z^2}.
\label{eq:seq-w-definition}
\end{equation}
For $w>0$, set
\begin{equation}
\begin{aligned}
T_c={}&
\frac{xy^2}{w}(E_{1,4}+E_{4,1})
+
\frac{xsz^2}{w}(E_{1,16}+E_{16,1})\\
&+
yz(E_{3,15}+E_{15,3})
+
\frac{z^2}{2}(E_{9,12}+E_{12,9}),
\end{aligned}
\label{eq:seq-achievability-Tc}
\end{equation}
and take $T_c=0$ when $w=0$. Here $E_{i,j}$ are matrix units in the
lexicographically ordered basis of $A_iA_oB_iB_o$. Direct substitution
verifies the $A\prec B$ comb constraints and
$\cD_{AB}(T_c)=0$. Positivity follows from
\begin{equation}
\left(\frac{xy}{w}\right)^2
+
\left(\frac{xsz}{w}\right)^2
=
x^2.
\label{eq:seq-achievability-positivity}
\end{equation}

Direct contraction with $\left(J_{\cN_{\rm AD}^{\varepsilon}}\right)^{\ox2}$ gives $x_c=F_\varepsilon(x,y,z)$, which is the achievability statement used in \cref{prop:sequential-main}.

\subsection{General-process construction}
For achievability, let $t\in[0,1]$ and define
\begin{equation}
\begin{aligned}
S(t)=\diag\bigg(&1-t,\frac{1-t}{2},\frac{t}{2},0,
\frac{1-t}{2},0,\frac12,\frac{1-t}{2},\\
&\frac{t}{2},\frac12,0,\frac{t}{2},
0,\frac{1-t}{2},\frac{t}{2},t\bigg),
\end{aligned}
\label{eq:ico-family-S}
\end{equation}
with $T_0(t)=S(t)$, $T_1(t)=0$, and
\begin{equation}
\begin{aligned}
T_c(t)={}&
\sqrt{t(1-t)}(E_{1,16}+E_{16,1})\\
&+
\frac{t}{2}
\left(
E_{3,15}+E_{15,3}
+E_{9,12}+E_{12,9}
\right).
\end{aligned}
\label{eq:ico-family-Tc}
\end{equation}
The family satisfies the general-process normalisation constraints, $\cD_{AB}(T_c)=0$, and $S(t)\pm T_c(t)\succeq0$.
The contraction of this family with the two amplitude-damping Choi operators gives \cref{eq:ico-family-xc} in the main text.

\section{Sequential upper-bound certificate}
\label{app:sequential-certificate}

For a computational-basis vector $\ket{a_i,a_o,b_i,b_o}$ define the two phase charges
\begin{equation}
q_A:=a_o-a_i,
\qquad
q_B:=b_o-b_i.
\label{eq:resource-phase-charges}
\end{equation}
The amplitude-damping Choi operator is invariant under $\overline U_\theta\ox U_\theta$ and has support only in the charge sectors $0$ and $-1$.

\begin{lemma}[Fixed-order charge-block reduction]
\label{lem:fixed-order-charge-blocks}
Every reduced $A\prec B$ MISC feasible point can be averaged, without changing $x_c$, over the independent phase rotations of the two resource channels and over entrywise complex conjugation. On the charge sectors that contribute to $x_c$, the resulting operators have the following principal blocks. In the ordered bases
\begin{equation}
\mathcal B_{00}=(\ket1,\ket4,\ket{13},\ket{16}),
\quad
\mathcal B_{0,-1}=(\ket3,\ket{15}),
\quad
\mathcal B_{-1,0}=(\ket9,\ket{12}),
\label{eq:seq-charge-bases}
\end{equation}
one has
\begin{equation}
S_{00}
=
\begin{pmatrix}
a&0&u&0\\
0&b&0&-u\\
u&0&g&0\\
0&-u&0&h
\end{pmatrix},
\quad
S_{0,-1}
=
\begin{pmatrix}b&-u\\-u&h\end{pmatrix},
\quad
S_{-1,0}
=
\begin{pmatrix}m_0&0\\0&n_0\end{pmatrix},
\label{eq:seq-charge-blocks}
\end{equation}
where
\begin{equation}
g+h=m_0+n_0,
\qquad
a+b+m_0+n_0=1.
\label{eq:seq-charge-relations}
\end{equation}
The corresponding blocks of $T_c$ are real symmetric and satisfy
\begin{equation}
C_{00}=C_{00}^{\mathsf T},
\quad
\diag C_{00}=0,
\quad
C_{0,-1}=K(r),
\quad
C_{-1,0}=K(t),
\qquad
K(q):=\begin{pmatrix}0&q\\q&0\end{pmatrix}.
\label{eq:Kq}
\end{equation}
Moreover,
\begin{equation}
x_c
=
\vect v^{\mathsf T}C_{00}\vect v+2\varepsilon s(r+t),
\qquad
\vect v=(1,s,s,s^2)^{\mathsf T}.
\label{eq:seq-block-objective}
\end{equation}
All matrices appearing as the second arguments of the certificate overlaps in \eqref{eq:seq-certificate-app} are principal blocks of $S\pm T_c$ and are therefore positive semidefinite.
\end{lemma}

\begin{proof}
The independent phase average removes matrix elements connecting different pairs $(q_A,q_B)$. Entrywise complex conjugation preserves positivity, the real causal equations, the real amplitude-damping objective, and the condition $\cD_{AB}(T_c)=0$, so a real representative may be chosen. The tensor-square Choi operator has nonzero entries contributing to the contraction with $T_c$ only on the sectors $(0,0)$, $(0,-1)$, and $(-1,0)$. The sector $(-1,-1)$ is one dimensional and cannot contribute because $T_c$ has zero diagonal.

Write the most general real charge-block-diagonal $S$ and impose, entry by entry,
\begin{equation}
S=\mathsf P_{B_o}(S),
\qquad
\tr_{B_iB_o}S=\mathsf P_{A_o}(\tr_{B_iB_o}S),
\qquad
\tr S=4.
\label{eq:seq-block-linear-system}
\end{equation}
Solving this finite linear system on the three contributing sectors gives exactly \eqref{eq:seq-charge-blocks}. The second equality in \eqref{eq:seq-block-linear-system} yields $g+h=m_0+n_0$, while the trace equation, after using the multiplicities enforced by $\mathsf P_{B_o}$, yields $a+b+m_0+n_0=1$. The same phase and reality reductions applied to $T_c$, together with $\cD_{AB}(T_c)=0$, give \eqref{eq:Kq}. Finally, restricting $G_\varepsilon=(J_{\cN_{\rm AD}^{\varepsilon}})^{\ox2}$ to the three sectors gives respectively $\vect v\vect v^{\mathsf T}$, $\varepsilon s\begin{psmallmatrix}0&1\\1&0\end{psmallmatrix}$, and the same $2\times2$ matrix, which proves \eqref{eq:seq-block-objective}. Principal submatrices of a positive semidefinite matrix are positive semidefinite, giving the final statement.
\end{proof}

We use $m_0,n_0$ for the entries in \eqref{eq:seq-charge-blocks} to avoid confusion with coherence coefficients.

\begin{lemma}[Interior maximiser and strict auxiliary inequalities]
\label{lem:seq-interior-properties}
Let $(x,y,z)$ maximise $F_\varepsilon$ over $x,y,z\ge0$ and $x^2+y^2+z^2=1$, and define
\begin{equation}
\mu:=F_\varepsilon(x,y,z).
\label{eq:mu-app-definition}
\end{equation}
For every $0<\varepsilon<1$,
\begin{equation}
x>0,
\qquad
y>0,
\qquad
z>0,
\qquad
\mu>s,
\qquad
z>y.
\label{eq:interior-properties}
\end{equation}
Moreover,
\begin{equation}
\Delta
:=
\mu-\varepsilon s-\frac{s^2}{\mu}
=
\varepsilon s\left(\frac zy-1\right)>0.
\label{eq:Delta-app}
\end{equation}
\end{lemma}

\begin{proof}
The feasible set is compact, so a maximiser exists. Suppose first that $x=0$. For sufficiently small $\delta>0$, set
\begin{equation}
(x_\delta,y_\delta,z_\delta)
=
\left(
\delta,
\sqrt{1-\delta^2}\,y,
\sqrt{1-\delta^2}\,z
\right).
\label{eq:x-perturbation}
\end{equation}
The terms independent of $x$ vary only at order $\delta^2$, whereas
\begin{equation}
\left.
\frac{\mathrm d}{\mathrm d\delta}
F_\varepsilon(x_\delta,y_\delta,z_\delta)
\right|_{\delta=0^+}
=
2s\sqrt{y^2+s^2z^2}>0.
\label{eq:x-derivative}
\end{equation}
Hence $x>0$.

If $y=0$ and $z>0$, use
\begin{equation}
(x_\delta,y_\delta,z_\delta)
=
\left(
\sqrt{1-\delta^2}\,x,
\delta,
\sqrt{1-\delta^2}\,z
\right).
\label{eq:y-perturbation}
\end{equation}
Then
\begin{equation}
\left.
\frac{\mathrm d}{\mathrm d\delta}
F_\varepsilon(x_\delta,y_\delta,z_\delta)
\right|_{\delta=0^+}
=
2\varepsilon s z>0.
\label{eq:y-derivative}
\end{equation}
If $y=z=0$, then $x=1$ and $F_\varepsilon=0$, while $F_\varepsilon(1/\sqrt2,1/\sqrt2,0)=s>0$. Thus $y>0$.

If $z=0$, then
\begin{equation}
F_\varepsilon(x,y,0)=2sxy\le s,
\label{eq:z-boundary}
\end{equation}
with equality only at $x=y=1/\sqrt2$. Along
\begin{equation}
x_\delta=y_\delta=\sqrt{\frac{1-\delta^2}{2}},
\qquad
z_\delta=\delta,
\label{eq:z-perturbation}
\end{equation}
one has
\begin{equation}
\left.
\frac{\mathrm d}{\mathrm d\delta}
F_\varepsilon(x_\delta,y_\delta,z_\delta)
\right|_{\delta=0^+}
=
\sqrt2\,\varepsilon s>0.
\label{eq:z-derivative}
\end{equation}
Therefore $z>0$, and the global maximum satisfies $\mu>s$.

The maximiser is interior, so the Lagrange equations are
\begin{equation}
\begin{aligned}
\mu x
&=s\sqrt{y^2+s^2z^2},\\
\mu y
&=\frac{sxy}{\sqrt{y^2+s^2z^2}}+\varepsilon sz,\\
\mu z
&=\frac{s^3xz}{\sqrt{y^2+s^2z^2}}+\varepsilon s(y+z).
\end{aligned}
\label{eq:seq-stationarity-app}
\end{equation}
The multiplier equals $\mu$ because $F_\varepsilon$ is homogeneous of degree two. Using the first equation in the second and third gives
\begin{equation}
\mu
=
\frac{s^2}{\mu}+\varepsilon s\frac zy,
\qquad
\mu
=
\frac{s^4}{\mu}+\varepsilon s\left(\frac yz+1\right).
\label{eq:stationarity-reduced}
\end{equation}
Subtracting and using $1-s^2=\varepsilon$ yields
\begin{equation}
\frac{s}{\mu}
=
1+\frac yz-\frac zy.
\label{eq:z-y-relation}
\end{equation}
Since $\mu>s$, the left-hand side is strictly smaller than one, so $y/z-z/y<0$. Because $y,z>0$, this is equivalent to $z>y$. The first identity in \cref{eq:stationarity-reduced} then gives \cref{eq:Delta-app}.
\end{proof}

Set
\begin{equation}
\vect n_\pm
:=
\begin{pmatrix}1\\\mp s/\mu\\0\\\mp s^2/\mu\end{pmatrix},
\qquad
\vect e_3:=\begin{pmatrix}0\\0\\1\\0\end{pmatrix},
\label{eq:n-pm}
\end{equation}
\begin{equation}
\begin{aligned}
A_\pm&:=\frac{\mu(1\mp\mu)}2,
& B_\pm&:=\frac{s(\mu\mp1)}2,\\
c_+&:=\frac{s^2(1-\mu)}{2\mu},
&c_-&:=\frac{s^2(1+\mu)}{2\mu}+\Delta,
\end{aligned}
\label{eq:ABC-pm}
\end{equation}
and
\begin{equation}
\begin{aligned}
P_\pm
&:=
A_\pm\vect n_\pm\vect n_\pm^{\mathsf T}
+B_\pm(\vect n_\pm\vect e_3^{\mathsf T}+\vect e_3\vect n_\pm^{\mathsf T})
+c_\pm\vect e_3\vect e_3^{\mathsf T},\\
Q_+
&:=
\alpha\begin{pmatrix}z\\-y\end{pmatrix}
\begin{pmatrix}z&-y\end{pmatrix},
\qquad
Q_-
:=
\beta\begin{pmatrix}z\\y\end{pmatrix}
\begin{pmatrix}z&y\end{pmatrix},\\
R_\pm
&:=
\frac{\varepsilon s}{2}
\begin{pmatrix}1\\\mp1\end{pmatrix}
\begin{pmatrix}1&\mp1\end{pmatrix},
\end{aligned}
\label{eq:seq-certificate-operators}
\end{equation}
where
\begin{equation}
\alpha
:=
\frac{\varepsilon s}{2yz}\left(1-s\frac zy\right),
\qquad
\beta
:=
\frac{\varepsilon s}{2yz}\left(1+s\frac zy\right).
\label{eq:alpha-beta}
\end{equation}

\begin{lemma}[Positivity of the sequential certificate]
\label{lem:seq-certificate-positivity}
The matrices above satisfy
\begin{equation}
P_\pm\succeq0,
\qquad
Q_\pm\succeq0,
\qquad
R_\pm\succeq0,
\qquad
P_-\succeq\Delta\vect e_3\vect e_3^{\mathsf T}.
\label{eq:certificate-positivity}
\end{equation}
\end{lemma}

\begin{proof}
The explicit sequential construction realises $F_\varepsilon(x,y,z)$ as the coherence coefficient of a CPTP qubit channel. Therefore $0<\mu\le1$. From the stationarity relation,
\begin{equation}
\mu-\frac{s^2}{\mu}
=
\varepsilon s\frac zy.
\label{eq:mu-stationarity}
\end{equation}
The function $f(r)=r-s^2/r$ is strictly increasing for $r>0$, since
\begin{equation}
f'(r)=1+\frac{s^2}{r^2}>0.
\label{eq:f-monotone}
\end{equation}
Using $\mu\le1$,
\begin{equation}
\varepsilon s\frac zy
=f(\mu)
\le f(1)
=1-s^2
=\varepsilon.
\label{eq:alpha-positive}
\end{equation}
Hence $sz/y\le1$, so $\alpha\ge0$, while $\beta>0$. Thus $Q_\pm\succeq0$; the rank-one form gives $R_\pm\succeq0$.

The vectors $\vect n_\pm$ and $\vect e_3$ are orthogonal, and
\begin{equation}
A_+c_+-B_+^2=0,
\qquad
A_-(c_--\Delta)-B_-^2=0.
\label{eq:P-determinants}
\end{equation}
Since $0<\mu\le1$, all diagonal coefficients of the corresponding $2\times2$ Gram matrices are nonnegative. Consequently,
\begin{equation}
P_+\succeq0,
\qquad
P_--\Delta\vect e_3\vect e_3^{\mathsf T}\succeq0.
\label{eq:P-positive}
\end{equation}
\end{proof}

\begin{lemma}[Sequential certificate identity]
\label{lem:seq-certificate-identity}
For every feasible block tuple in \cref{lem:fixed-order-charge-blocks},
\begin{equation}
\begin{aligned}
\mu-x_c={}&
\langle P_+,S_{00}+C_{00}\rangle
+
\langle P_-,S_{00}-C_{00}\rangle\\
&+
\langle Q_+,S_{0,-1}+K(r)\rangle
+
\langle Q_-,S_{0,-1}-K(r)\rangle\\
&+
\langle R_+,S_{-1,0}+K(t)\rangle
+
\langle R_-,S_{-1,0}-K(t)\rangle.
\end{aligned}
\label{eq:seq-certificate-app}
\end{equation}
\end{lemma}

\begin{proof}
Expand the right-hand side in the independent real variables consisting of the six off-diagonal entries of $C_{00}$, the two variables $r,t$, and the entries of the three $S$ blocks subject to \eqref{eq:seq-charge-relations}. By the definitions of $\vect n_\pm$, the coefficient of $(C_{00})_{ij}$ is $-2v_iv_j$ for $i<j$. The rank-one forms of $Q_\pm$ and $R_\pm$ give the coefficients $-2\varepsilon s$ of both $r$ and $t$. These terms therefore sum to $-x_c$ by \eqref{eq:seq-block-objective}. For the $S$ variables, substitute \eqref{eq:seq-charge-relations} and then use the three stationarity equations \eqref{eq:seq-stationarity-app} together with \eqref{eq:Delta-app}. The coefficients of every free $S$ variable vanish, while the constant normalisation term is $\mu$. Thus the expansion is exactly $\mu-x_c$.
\end{proof}

Every second argument in \cref{eq:seq-certificate-app} is positive semidefinite by \cref{lem:fixed-order-charge-blocks}, while every first argument is positive semidefinite by \cref{lem:seq-certificate-positivity}. Therefore $x_c\le\mu$, proving the converse part of \cref{prop:sequential-main}.

For completeness, the strict parallel exclusion used in \cref{lem:no-parallel-saturation-main} follows without any hidden positivity assumption. Equality in \cref{eq:seq-certificate-app} implies
\begin{equation}
0
=
\langle P_-,S_{00}-C_{00}\rangle
\ge
\Delta\,\langle e_3,(S_{00}-C_{00})e_3\rangle
=
\Delta g,
\label{eq:g-zero}
\end{equation}
so $g=0$. Positivity of $\diag(g,h)\pm K(t)$ gives $t=0$, and then
$\langle R_+,\diag(0,h)\rangle=\varepsilon s h/2=0$ gives $h=0$.
Next, positivity of $\diag(b,0)\pm K(r)$ gives $r=0$, and
\begin{equation}
0
=
\langle Q_-,\diag(b,0)\rangle
=
\beta z^2b
\label{eq:b-zero}
\end{equation}
gives $b=0$. Normalisation gives $a=1$. Finally, a positive semidefinite
matrix with a zero diagonal entry has a zero row and column at that entry.
Applying this fact to both $S_{00}+C_{00}$ and $S_{00}-C_{00}$ shows
$C_{00}=0$, and hence $x_c=0$, contradicting $\mu\ge s>0$.

\section{General-process upper-bound certificate}
\label{app:ico-certificate}

\begin{lemma}[General-process charge-block reduction]
\label{lem:ico-charge-blocks}
For every reduced general-process feasible point with $T_1=0$, averaging over the independent amplitude-damping phase symmetries, entrywise complex conjugation, and $A\leftrightarrow B$ preserves feasibility and $x_c$. In the bases $(\ket1,\ket4,\ket{13},\ket{16})$, $(\ket3,\ket{15})$, and $(\ket9,\ket{12})$, the contributing blocks of $S$ are
\begin{equation}
\begin{pmatrix}
r_0-2r_1+2r_4&-r_2&-r_2&0\\
-r_2&r_1-r_3+r_5&0&r_2\\
-r_2&0&r_1-r_3+r_5&r_2\\
0&r_2&r_2&2r_3-r_5
\end{pmatrix},
\quad
\begin{pmatrix}r_1&r_2\\r_2&r_3\end{pmatrix},
\quad
\begin{pmatrix}r_1&r_2\\r_2&r_3\end{pmatrix},
\label{eq:ico-blocks-app}
\end{equation}
with
\begin{equation}
r_0+2r_4+r_5=1.
\label{eq:ico-normalisation-app}
\end{equation}
All omitted charge sectors are orthogonal to the off-diagonal support of $G_\varepsilon$ and therefore do not contribute to $x_c$.
\end{lemma}

\begin{proof}
The phase and reality arguments are the same as in \cref{lem:fixed-order-charge-blocks}. The exchange average is allowed because the two resources are identical and the general-process equations, the free constraint, and $x_c$ are invariant under $A\leftrightarrow B$. Write a general real operator on the three contributing charge sectors and impose
\begin{equation}
\begin{aligned}
S&=\mathsf P_{A_o}(S)+\mathsf P_{B_o}(S)-\mathsf P_{A_oB_o}(S),\\
\tr_{A_iA_o}S&=\mathsf P_{B_o}(\tr_{A_iA_o}S),\\
\tr_{B_iB_o}S&=\mathsf P_{A_o}(\tr_{B_iB_o}S),
\qquad
\tr S=4,
\end{aligned}
\label{eq:ico-block-linear-system}
\end{equation}
together with exchange symmetry. Solving this finite linear system gives \eqref{eq:ico-blocks-app}; its remaining independent trace equation is exactly \eqref{eq:ico-normalisation-app}. The tensor-square amplitude-damping Choi operator has off-diagonal support only on the displayed sectors. Since $\cD_{AB}(T_c)=0$, one-dimensional sectors contribute neither to $x_c$ nor to the certificate pairing with $T_c$.
\end{proof}
Let
\begin{equation}
a:=\sqrt{1-t_\star},
\qquad
b:=\sqrt{t_\star}.
\label{eq:a-b-ico}
\end{equation}
The definition of $t_\star$ gives
\begin{equation}
\frac ba=\frac{R_\varepsilon+\varepsilon}{s},
\qquad
\frac ab=\frac{R_\varepsilon-\varepsilon}{s},
\qquad
ab=\frac{s}{2R_\varepsilon}.
\label{eq:a-b-relations}
\end{equation}
Define
\begin{equation}
\ket{e_\pm}:=\frac{\ket4\pm\ket{13}}{\sqrt2},
\quad
\ket{w_P}:=-b\ket1+a\ket{16},
\quad
\ket{w_Q}:=b\ket1+a\ket{16},
\label{eq:ico-vectors-1}
\end{equation}
\begin{equation}
\ket{\phi_\pm^A}:=\frac{\ket3\pm\ket{15}}{\sqrt2},
\qquad
\ket{\phi_\pm^B}:=\frac{\ket9\pm\ket{12}}{\sqrt2},
\label{eq:ico-vectors-2}
\end{equation}
and
\begin{equation}
\begin{aligned}
\lambda&:=\frac{\sqrt2a}{s},\\
\alpha&:=\frac{s^2[1-s(\varepsilon+R_\varepsilon)]}{2ab},
\qquad
\beta:=\frac{s^2[1+s(\varepsilon+R_\varepsilon)]}{2ab},\\
\rho&:=\frac{\varepsilon(1-R_\varepsilon)}s,
\qquad
\kappa:=\varepsilon(2-\varepsilon-R_\varepsilon).
\end{aligned}
\label{eq:ico-certificate-coefficients}
\end{equation}
Set
\begin{equation}
\begin{aligned}
P_\varepsilon={}&
\alpha(\ket{w_P}+\lambda\ket{e_+})
(\bra{w_P}+\lambda\bra{e_+})
+\rho\ketbra{e_-}{e_-}\\
&+(\varepsilon s-\kappa)
\left(
\ketbra{\phi_-^A}{\phi_-^A}
+
\ketbra{\phi_-^B}{\phi_-^B}
\right),\\
Q_\varepsilon={}&
\beta(\ket{w_Q}+\lambda\ket{e_+})
(\bra{w_Q}+\lambda\bra{e_+})
+\rho\ketbra{e_-}{e_-}\\
&+(\varepsilon s+\kappa)
\left(
\ketbra{\phi_+^A}{\phi_+^A}
+
\ketbra{\phi_+^B}{\phi_+^B}
\right).
\end{aligned}
\label{eq:ico-PQ}
\end{equation}

Positivity of the certificate operators follows directly from their coefficients.
Since
\begin{equation}
R_\varepsilon^2=1-\varepsilon(1-\varepsilon),
\label{eq:R-square}
\end{equation}
one has $R_\varepsilon\le1$. Moreover,
\begin{equation}
(1-\varepsilon s)^2-s^2R_\varepsilon^2
=
\varepsilon(1-s)^2\ge0,
\label{eq:m-less-one-proof}
\end{equation}
and $1-\varepsilon s>0$, so
$s(\varepsilon+R_\varepsilon)\le1$. Finally,
\begin{equation}
R_\varepsilon^2-(2-\varepsilon-s)^2
=
2s(1-s)^2\ge0
\label{eq:kappa-upper-proof}
\end{equation}
implies $2-\varepsilon-R_\varepsilon\le s$ and hence
$\kappa\le\varepsilon s$; also $\kappa\ge0$ follows from
$R_\varepsilon\le1$. Therefore
\begin{equation}
\alpha,\beta,\rho,\varepsilon s-\kappa,\varepsilon s+\kappa\ge0,
\label{eq:ico-nonnegative-coefficients}
\end{equation}
and consequently $P_\varepsilon,Q_\varepsilon\succeq0$.

Let
\begin{equation}
G_\varepsilon
:=
\left(J_{\cN_{\rm AD}^{\varepsilon}}\right)^{\ox2}.
\label{eq:G-epsilon-app}
\end{equation}
Using \cref{eq:a-b-relations}, the off-diagonal entries of
$Q_\varepsilon-P_\varepsilon$ on the first charge block are
\begin{equation}
\left(Q_\varepsilon-P_\varepsilon\right)_{(1,4,13,16)}
=
\begin{pmatrix}
\ast&s&s&s^2\\
s&\ast&s^2&s^3\\
s&s^2&\ast&s^3\\
s^2&s^3&s^3&\ast
\end{pmatrix},
\label{eq:QP-first-block-offdiag}
\end{equation}
where the asterisks denote diagonal entries. On each of the other two
charge blocks,
\begin{equation}
\left(Q_\varepsilon-P_\varepsilon\right)_{(3,15)}
=
\left(Q_\varepsilon-P_\varepsilon\right)_{(9,12)}
=
\begin{pmatrix}
\ast&\varepsilon s\\
\varepsilon s&\ast
\end{pmatrix}.
\label{eq:QP-other-blocks-offdiag}
\end{equation}
These are exactly all off-diagonal entries of $G_\varepsilon$ in the same
basis. Hence
\begin{equation}
H_\varepsilon
:=
G_\varepsilon-(Q_\varepsilon-P_\varepsilon)
\label{eq:H-epsilon-definition}
\end{equation}
is diagonal. Since $\cD_{AB}(T_c)=0$,
\begin{equation}
\langle Q_\varepsilon-P_\varepsilon,T_c\rangle
=
\langle G_\varepsilon,T_c\rangle
=
x_c.
\label{eq:ico-QP-Tc}
\end{equation}

The remaining ingredient is the normalization identity. Put
\begin{equation}
m_\varepsilon:=s(\varepsilon+R_\varepsilon).
\label{eq:m-epsilon-app}
\end{equation}
The entries of $P_\varepsilon+Q_\varepsilon$ that contribute to the first
block in \cref{eq:ico-blocks-app} are
\begin{equation}
\begin{aligned}
&(P_\varepsilon+Q_\varepsilon)_{11}=m_\varepsilon,\\
&(P_\varepsilon+Q_\varepsilon)_{22}
=(P_\varepsilon+Q_\varepsilon)_{33}=R_\varepsilon s,\\
&(P_\varepsilon+Q_\varepsilon)_{44}
=s(R_\varepsilon-\varepsilon),\\
&(P_\varepsilon+Q_\varepsilon)_{12}
=(P_\varepsilon+Q_\varepsilon)_{13}
=s^2(R_\varepsilon+\varepsilon),\\
&(P_\varepsilon+Q_\varepsilon)_{24}
=(P_\varepsilon+Q_\varepsilon)_{34}
=R_\varepsilon-\varepsilon.
\end{aligned}
\label{eq:PQ-relevant-entries}
\end{equation}
On each $2\times2$ block,
\begin{equation}
(P_\varepsilon+Q_\varepsilon)_{(3,15)}
=
(P_\varepsilon+Q_\varepsilon)_{(9,12)}
=
\begin{pmatrix}
\varepsilon s&\kappa\\
\kappa&\varepsilon s
\end{pmatrix}.
\label{eq:PQ-small-blocks}
\end{equation}
Taking the Hilbert--Schmidt inner product with
\cref{eq:ico-blocks-app} gives
\begin{equation}
\begin{aligned}
\langle P_\varepsilon+Q_\varepsilon,S\rangle
={}&
m_\varepsilon(r_0-2r_1+2r_4)
+2R_\varepsilon s(r_1-r_3+r_5)\\
&+s(R_\varepsilon-\varepsilon)(2r_3-r_5)
+4\left[
R_\varepsilon-\varepsilon
-s^2(R_\varepsilon+\varepsilon)
\right]r_2\\
&+2\varepsilon s(r_1+r_3)+4\kappa r_2.
\end{aligned}
\label{eq:PQ-S-expanded}
\end{equation}
Now
\begin{equation}
R_\varepsilon-\varepsilon
-s^2(R_\varepsilon+\varepsilon)
=
-\varepsilon(2-\varepsilon-R_\varepsilon)
=
-\kappa,
\label{eq:r2-cancellation}
\end{equation}
so the $r_2$ terms cancel. The coefficients of $r_1$ and $r_3$ also
vanish:
\begin{equation}
-2m_\varepsilon+2R_\varepsilon s+2\varepsilon s=0,
\qquad
-2R_\varepsilon s+2s(R_\varepsilon-\varepsilon)+2\varepsilon s=0.
\label{eq:r1-r3-cancellation}
\end{equation}
The remaining terms are
\begin{equation}
\langle P_\varepsilon+Q_\varepsilon,S\rangle
=
m_\varepsilon(r_0+2r_4+r_5)
=
m_\varepsilon,
\label{eq:ico-PQ-S}
\end{equation}
where the last equality uses \cref{eq:ico-normalisation-app}. Combining
\cref{eq:ico-QP-Tc,eq:ico-PQ-S} gives
\begin{equation}
\begin{aligned}
m_\varepsilon-x_c
&=
\langle P_\varepsilon,S+T_c\rangle
+
\langle Q_\varepsilon,S-T_c\rangle\\
&\ge0,
\end{aligned}
\label{eq:ico-certificate-app}
\end{equation}
because $P_\varepsilon,Q_\varepsilon\succeq0$ and
$S\pm T_c\succeq0$.

\section{Parallel phase twirl and DISC feasibility}
\label{app:parallel-disc-twirl}

For a single amplitude-damping channel, let
\begin{equation}
U_\theta:=\ketbra{0}{0}+e^{i\theta}\ketbra{1}{1}.
\label{eq:phase-unitary}
\end{equation}
Its Choi operator satisfies
\begin{equation}
(\overline U_\theta\ox U_\theta)
J_{\cN_{\rm AD}^{\varepsilon}}
(\overline U_\theta\ox U_\theta)^\dagger
=
J_{\cN_{\rm AD}^{\varepsilon}}.
\label{eq:ad-phase-covariance}
\end{equation}
Apply the corresponding conjugation independently on the $A$ and $B$ input-channel registers of a reduced parallel feasible point and average over both phases. The link-product objective is invariant by \cref{eq:ad-phase-covariance}. Parallel normalisation with $T_1=0$ gives
\begin{equation}
T_0=\rho_{A_iB_i}\ox I_{A_oB_o}.
\label{eq:parallel-rho-form}
\end{equation}
The independent phase average sends $\rho_{A_iB_i}$ to $\cD_{A_iB_i}(\rho_{A_iB_i})$ and leaves $I_{A_oB_o}$ unchanged, so the averaged $T_0$ is fully diagonal on $A_iA_oB_iB_o$. The same average preserves $\cD_{AB}(T_c)=0$ because complete dephasing commutes with the phase action. Positivity of $T_0\pm T_c$ and the objective are preserved by convexity. The averaged point is therefore DISC feasible with the same $x_c$.

\section{Numerical MISC--DISC separation example}
\label{app:misc-disc-example}

We define the channel pair used in \cref{tab:misc-disc-numerical} by the exact Choi operators
\begin{equation}
J_{\cN_1}:=10^{-4}K_1,
\qquad
J_{\cN_2}:=10^{-4}K_2,
\label{eq:exact-numerical-channel-definition}
\end{equation}
where, in the input--output basis order $\{\ket{00},\ket{01},\ket{10},\ket{11}\}$,
\begin{equation}
\resizebox{0.96\linewidth}{!}{$
K_1
=
\begin{pmatrix}
5717&-947+1220i&-474+567i&592+537i\\
-947-1220i&4283&-45+2150i&474-567i\\
-474-567i&-45-2150i&6068&-755-2755i\\
592-537i&474+567i&-755+2755i&3932
\end{pmatrix}
$}.
\label{eq:K1-exact}
\end{equation}
\begin{equation}
\resizebox{0.96\linewidth}{!}{$
K_2
=
\begin{pmatrix}
4828&-1180+428i&-511+630i&1102+68i\\
-1180-428i&5172&115-1401i&511-630i\\
-511-630i&115+1401i&2717&-660-5i\\
1102-68i&511+630i&-660+5i&7283
\end{pmatrix}
$}.
\label{eq:K2-exact}
\end{equation}
All entries of $K_1$ and $K_2$ are exact Gaussian integers. Thus the four-decimal entries of $J_{\cN_1}$ and $J_{\cN_2}$ are themselves the exact numerical definition of the revised instance; they are not rounded representations of a different full-precision data file.

The trace-preservation condition is exact. Taking the partial trace over the channel output gives
\begin{equation}
\tr_{X_o}K_j=10^4I_{X_i},
\qquad
\tr_{X_o}J_{\cN_j}=I_{X_i},
\qquad
j=1,2.
\label{eq:exact-channel-TP-check}
\end{equation}
For example, for $K_1$ the two diagonal entries of the partial trace are
\begin{equation}
5717+4283=10^4,
\qquad
6068+3932=10^4,
\end{equation}
and its off-diagonal entry is
\begin{equation}
(-474+567i)+(474-567i)=0.
\end{equation}
The same identities hold for $K_2$.

Positivity can also be certified without floating-point eigenvalues. Let $K_j[1\!:\!k]$ denote the leading $k\times k$ principal submatrix. The leading principal minors are
\begin{equation}
\begin{array}{c|rr}
\toprule
k & \det K_1[1\!:\!k] & \det K_2[1\!:\!k]\\
\midrule
1 & 5717 & 4828\\
2 & 22100702 & 23394832\\
3 & 105404493296 & 52291109988\\
4 & 257783065977802 & 347392338883277\\
\bottomrule
\end{array}
\label{eq:exact-channel-leading-minors}
\end{equation}
All of these integers are strictly positive. Since $K_1$ and $K_2$ are Hermitian, Sylvester's criterion implies
\begin{equation}
K_1\succ0,
\qquad
K_2\succ0,
\end{equation}
and therefore \cref{eq:exact-numerical-channel-definition} defines two CPTP qubit channels exactly.

Solving the six reduced transformation SDPs for the exact pair in \cref{eq:exact-numerical-channel-definition}, with the identity channel as target, gives the values reported in \cref{tab:misc-disc-numerical}. Those displayed decimals are numerical solver outputs rounded to seven decimal places and are not claimed to be exact algebraic constants.

\section{Teleportation stretching in explicit form}
\label{app:teleportation-details}

Let
\begin{equation}
\ket{\Phi_\alpha}:=(I\ox\sigma_\alpha)\ket{\Phi^+},
\qquad
\ket{\Phi^+}:=\frac{\ket{00}+\ket{11}}{\sqrt2}.
\label{eq:bell-states-app}
\end{equation}
For a mixed-Pauli channel,
\begin{equation}
\omega_{\cN_{\vect p}}
=
\sum_{\alpha=0}^3p_\alpha\ketbra{\Phi_\alpha}{\Phi_\alpha}.
\label{eq:pauli-choi-bell}
\end{equation}
Teleportation through $\ket{\Phi_\alpha}$ implements
$\rho\mapsto\sigma_\alpha\rho\sigma_\alpha^\dagger$ up to an irrelevant Pauli phase. Linearity proves
\begin{equation}
\mathfrak T_{\rm tel}(\cN_{\vect p})=\cN_{\vect p}.
\label{eq:pauli-fixed-point-app}
\end{equation}

For an arbitrary channel tuple $(\cE_1,\ldots,\cE_n)$, the parallelised protocol prepares $\omega_{\cE_1}\ox\cdots\ox\omega_{\cE_n}$ by applying every $\cE_k$ to half of a maximally entangled state. These states are stored and inserted into the fixed teleportation operations wherever the original higher-order protocol calls its resource channels. This is the operational content of \cref{eq:parallelised-superchannel}.

\end{document}